\documentclass[12pt]{iopart}
\usepackage{epsfig}
\usepackage{amsthm}
\usepackage{graphicx}
\usepackage{color}
\usepackage{xcolor}
\usepackage{hyperref}
\usepackage{bbm}
\usepackage{soul}
\usepackage{ulem}
\usepackage{algorithm}
\usepackage{algpseudocode}
\usepackage{tikz}
\usepackage{cite}

\def\id{{\rm I}}
\newcommand{\de}[1]{\left( #1 \right)}
\newcommand{\De}[1]{\left[#1\right]}
\newcommand{\DE}[1]{\left\{#1\right\}} 
\newcommand{\ket}[1]{\left| #1 \right\rangle}

\newcommand{\ketbra}[2]{\left|#1\middle\rangle\middle\langle#2\right|}

\newcommand{\beq}{\begin{equation}}
\newcommand{\eeq}{\end{equation}}

\newtheorem{definition}{Definition}
\newtheorem{lemma}{Lemma}

\newtheorem{assumptions}{Assumptions}
\newtheorem{prop}{Properties}
\newtheorem*{thm6}{Theorem \ref{thm:H2}}

\newcommand{\ie}{\textit{i.e.}}

\newcommand{\ep}{\epsilon}

\newcommand{\gm}[1]{\textcolor{black}{#1}}

\newtheorem{theorem}{Theorem}
\begin{document}

\title{Towards a realization of device-independent quantum key distribution}

\author{G. Murta$^1$, S.B. van Dam$^{1,2}$,
J. Ribeiro$^1$, R. Hanson$^{1,2}$, and S. Wehner$^1$}
\address{$^1$ QuTech, Delft University of Technology, Lorentzweg 1, 2628 CJ Delft, The Netherlands}
\address{$^2$ Kavli Institute of Nanoscience, Delft University of Technology, Lorentzweg 1, 2628 CJ Delft, The Netherlands}
\ead{glauciamg.fis@gmail.com}

\begin{abstract}
In the implementation of device-independent quantum key distribution we are interested in maximizing the key rate, i.e. the number of key bits that can be obtained per signal, for a fixed security parameter. In the finite size regime, we furthermore also care about the minimum number of signals required before key can be obtained at all. Here, we perform a fully finite size analysis of device independent protocols using the CHSH inequality both for collective and coherent attacks. For coherent attacks, we sharpen the results recently derived in Arnon-Friedman \textit{et al.}, Nat. Commun. 9, 459 (2018) \cite{EATpublish}, to reduce the minimum number of signals before key can be obtained. In the regime of collective attacks, where the devices are restricted to have no memory, we employ two different techniques that exploit this restriction to further reduce the number of signals. 
We then discuss experimental platforms in which DIQKD may be implemented. We analyse Bell violations and expected QBER achieved in previous Bell tests with distant setups and situate these parameters
in the security analysis. Moreover, focusing on one of the experimental platforms, namely nitrogen-vacancy based systems, we describe experimental improvements that can lead to a device-independent quantum key distribution implementation in the near future.	
\end{abstract}


\maketitle

\section{Introduction}

\subsection{Quantum key distribution}

Quantum key distribution (QKD) \cite{BB84, E91} is a remarkable example of the advantages that quantum systems bring to accomplishing  classical tasks.
All the classical crypto-systems used for key exchange are based on computational assumptions and, therefore, are susceptible to retroactive attacks.  
 Indeed, if an adversary keeps track of the public information exchanged during the communication of an encrypted message and, in a later future, a more efficient algorithm or faster machines become available, then the messages exchanged in the past can be decrypted.
The novelties brought by quantum systems allow two parties to establish a common key that is  information-theoretically secure and, therefore, can be used to achieve perfect secure communication with a one-time pad encryption. 

Quantum key distribution schemes explore intrinsic properties of quantum systems, such as no-cloning \cite{Dieks82,WZ82} and monogamy of entanglement \cite{CKW00}, in order to achieve security even against an all powerful adversary who has unlimited computational power.  
The well known quantum key distribution scheme BB84 \cite{BB84}
can tolerate a reasonable amount of noise and decent rates\footnote{Due to finite size effects a minimal number of rounds is required in order to guarantee security. For the BB84 protocol this minimal number of rounds required is $\sim 10^4$. Moreover, a quantum bit error rate (QBER) of up to $20\%$ can be tolerated \cite{GL03,Chau02} for large enough number of rounds.} can be achieved with current technology, see for example the analyses of \cite{TLGR12,HT12,RigorousQKD}. BB84-based QKD has been successfully implemented over long distances, see for example \cite{Hiskett06,Korzh15},
and even satellite-based secure quantum communication was established \cite{Liao18}.

A successful implementation of the BB84 protocol is, however, highly dependent on a good characterisation of the underlying quantum system and the measurement devices. For example, the protocol can easily be broken if the devices are performing measurements in four dimensional systems instead of qubits, see discussions in \cite{MayersYao98, PAB09}. Furthermore, hacking of existent implementations that exploit experimental imperfections were presented (see e.g. \cite{Zhao08,Lydersen,Weier,Gerhardt11}).

A good characterization of the experimental setup is a strong assumption. What is more, when quantum technologies become commercially  available, we might often  buy devices from a provider which is not entirely trustworthy. 
Fortunately, quantum properties allow us to overcome this problem: By exploring the strong correlations that arise in quantum systems, one can prove security of quantum key distribution even in the very adversarial scenario where Alice and Bob  \gm{do not have complete knowledge of the internal working} 
of their measurement devices  or the underlying quantum system that they are measuring \cite{PAB09,EATpublish,AGM06,AMP06,SGB+06,Mas09,HRW10,PMLA13,PAB07,MPA11,HR10,MRC+14,BCK12b,RUV12,VV14,MS14}. This is the \emph{device-independent} (DI) model. 

\subsection{The device-independent scenario}

The device-independent scenario models the underlying system and measurement devices as black boxes where the only relevant information is the statistics of inputs and outputs. Therefore, no assumptions on the dimension of the quantum systems or the particular measurements performed by the devices are required. 
This represents a significant relaxation of the assumptions present in an implementation of the BB84 protocol. However, it is important to remark which assumptions remain present in any implementation of a DI protocol.

\begin{assumptions}[Device-independent model]\label{assump:DI}
In the device-independent model we assume:
\begin{enumerate}
\item\label{assump:labs} Isolated labs: no information is leaked from or enters Alice's and Bob's labs, apart from the state distribution before the measurements and the public classical information dictated by the protocol.
\item\label{assump:source} Isolated source: the preparation of states is independent of the measurements.
\item\label{assump:authentication} Trusted classical post-processing: all the public classical communication is performed using an authenticated channel and the local classical computations are trusted.
\item\label{assump:RNG} Trusted Random Number Generators: Alice and Bob possess independent and trusted random number generators. 
\end{enumerate}
\end{assumptions}

A bit of thought can make one conclude that completely removing any of these assumptions leads to a strategy where the key is leaked to the adversary. However, we remark that partial relaxation of these assumptions can still be considered. In Ref.~\cite{Unruh13}, QKD is proved to achieve everlasting security by relaxing Assumption~\ref{assump:DI}\textit{(\ref{assump:authentication})}
 to a computationally secure authenticated channel, but assuming the eavesdropper to be computationally bounded during the execution of the protocol.
 In many device independent protocols, instead of Assumption~\ref{assump:DI}\textit{(\ref{assump:source})}, it is assumed that all the $n$ systems are prepared before the measurement phase starts, so that no information other than the classical public communication is exchanged during the protocol. However, this would require quantum memory from Alice and Bob in order to store the quantum states along the protocol. In an implementation where the quantum states are generated round by round, and therefore in which no long term quantum  memory is required,
  Assumption~\ref{assump:DI}\textit{(\ref{assump:source})} is necessary to avoid that the state prepared by the source leaks the raw bits generated by Alice's device in the previous round. Indeed, if the source is arbitrarily  correlated with the measurement devices the state prepared can contain an additional degree of freedom that encodes the string of bits generated in the previous rounds (this strategy is detailed in \cite[Appendix C]{2partyJeremy}). 
We remark that, in experimental platforms, 
the preparation of states and the measurements are either performed 
within the same systems or optically connected ones, and therefore one needs to assume that the process of generating a quantum state is not correlated with the previously performed measurements. This assumption is, however, often well justified based on a  description of the setup.
\gm{Ref.~\cite{BCK13} addresses the problem of hidden memory in the devices. The authors show that a malicious eavesdropper can program the measurement devices in such a way that information about a previously generated key  may be leaked through the public communication of a subsequent run of the key generation protocol, if the devices are re-used. 
Ref.~\cite{CL19} proposes an alternative to overcome memory attacks and covert channels in general, as well as the need to assume that all the classical post-processing is trusted. By introducing protocols based on secure multi-party computation distributed among more devices, ref.~\cite{CL19} relaxes the black-box model to reliability of only one of the quantum devices. Moreover, the classical post-processing can tolerate up to a third of malicious classical devices. }
  \vspace{1em}

Another assumption that is often used in security proofs is that the rounds of the experiment are \textit{independent and identically distributed} (IID). This, in particular, implies that the measurement devices are memoryless and the state shared by Alice and Bob is the same for every round on the protocol. The IID assumption can be justified, for example, in experimental setups where Alice and Bob control to some extent the source and measurement devices, but do not have a full characterization of their working.

\begin{assumptions}[IID assumption]\label{assump:iid}
An IID implementation assumes:
\begin{itemize}
\item IID devices: the devices behave independently and in the same way in every round of the protocol. 
\item IID states: The state distributed is the same for every round of the protocol. In summary, the state of the $n$ rounds can be written as $\rho_{A_1^nB_1^nE}=\rho_{ABE}^{\otimes n}$. 
\end{itemize}
\end{assumptions}

 The eavesdropper attacks in QKD are classified in three types: \emph{Individual attacks}, where the eavesdropper has no memory and therefore is restricted to attack individually each round of the protocol; \emph{Collective attacks:} where in every round the systems of Alice and Bob, as well as the measurement devices, are prepared identically but the eavesdropper is allowed to make arbitrary global operations on her quantum side information; and  \emph{Coherent attacks:} additionally to the global operations the eavesdropper can perform in her quantum side information,  the states shared by Alice and Bob in each round can be arbitrarily correlated, as well as the measurement devices in the DI scenario
 can have memory and operate according to the results of previous rounds, {\textit{i.e.}, do not satisfy  the IID assumption.}
The IID assumption, stated in Assumptions~\ref{assump:iid}, corresponds to the scenario where the eavesdropper is restricted to collective attacks.
In what follows we focus on two types of adversarial attacks: collective attacks and coherent attacks.

\subsection{Device-independent quantum key distribution protocols}

The first ideas of device-independent QKD arose in the E91 protocol \cite{E91}, which uses a test of the CHSH inequality \cite{CHSH}  in order to certify that Alice and Bob share a maximally entangled state. This idea of self-testing quantum devices was further explored in \cite{MayersYao98}. Indeed, 
device-independent quantum key distribution relies on the violation of a Bell inequality in order to certify the security of the generated key. The simplest DIQKD protocol uses the CHSH inequality  for the security test:
	\begin{equation}
	\beta=\left<A_0B_0\right>+\left<A_0B_1\right>+\left<A_1B_0\right>-\left<A_1B_1\right>\leq 2, 
	\end{equation}
	where $\left<A_xB_y\right>=p(a=b|xy)-p(a\neq b|xy)$ represents the correlation of the outputs $a,b$ of Alice and Bob when they perform the measurement labeled by $x,y$ respectively. The CHSH inequality can be phrased  as a game \cite{CHTW04} in which Alice and Bob receive $x$ and $y$, respectively, as inputs and the winning condition is that their outputs satisfy  $a+b=x\cdot y$, with the operations $+,\cdot$ taken  modulo 2. The winning probability $\omega$ of the CHSH game relates to the violation $\beta$ by 
	\begin{equation}
	\omega=\frac{4+\beta}{8}.
	\end{equation}

For DIQKD based on the CHSH inequality, we consider protocols where Alice possesses a device with two possible inputs $X\in\DE{0,1}$ and Bob has a device with three possible inputs $Y \in \DE{0,1,2}$. The inputs $X\in\DE{0,1}$ and $Y \in \DE{0,1}$
	are used to test for the CHSH inequality, and the inputs $X=0$ and $Y=2$ are used for the other rounds, often called key generation rounds, where maximal correlation of the outputs is expected.
	The parameters of interest are the Bell violation $\beta$, or winning probability $\omega$, achieved in the test rounds and the quantum bit error rate (QBER) $Q$ of the key generation rounds.
	 We consider that an implementation of the protocol is expected to have $n$ rounds and a portion $\gamma n$ of these rounds is used for testing of the CHSH condition.
	
	A DIQKD protocol can be divided in three phases:
	\begin{itemize}
		\item An initial phase where Alice and Bob use their respective devices to measure the quantum systems and, according to the obtained outputs, generate the $n$-bit strings $A_1^n$ and $B_1^n$.
		\item  A second phase where Alice and Bob publicly exchange classical information in order to perform \emph{error correction}, to correct their respective strings generating the raw keys; and \emph{parameter estimation}, to estimate the parameters of interest (Bell violation, $\beta$, and QBER, $Q$). At the end of this phase Alice and Bob are supposed to share equal $n$-bit strings and have an estimate of how much knowledge an eavesdropper might have about their raw key.
		\item In the final phase, Alice and Bob perform \emph{privacy amplification}, where the not fully secure $n$-bit strings are mapped into smaller strings $K_A$ and $K_B$, which represents the final keys of Alice and Bob respectively.
\end{itemize}	
	
	The specific protocols we consider for our analyses are detailed in Section~\ref{sec:results}, (see Protocol \ref{prot:diqkd} and Protocol \ref{prot:diqkdIID}).
	
	In order to define security of a DIQKD protocol, we follow Refs.~\cite{DIEAT,EATpublish} and adopt the security definition that is universally composable for standard QKD protocols \cite{PR14}. Universal composability is the statement that a protocol remains secure even if it is used arbitrarily in composition with other protocols.
It is important to remark that, for the device-independent case,  attacks proposed in Ref.~\cite{BCK13} show that composability is not achieved if the same devices are re-used for generation of a subsequent key. Indeed, in \cite{BCK13},  the authors have shown that a malicious eavesdropper can program the measurement devices in such a way that information about a previously generated key  may be leaked through the public communication of a subsequent run of the key generation protocol, if the devices are re-used. It is still an open problem what is the minimum set of assumptions that can lead to universal composability of DIQKD (e.g. the attacks of Ref.~\cite{BCK13} can be avoided if \gm{we assume that Alice and Bob have sufficient control over the existing internal memory of their devices, so that they can re-set it} 
after an execution of the protocol).

Let $K_A$ and $K_B$ denote the final key held by Alice and Bob, respectively, after they perform a DIQKD protocol.	
	A DIQKD protocol is secure if it is \emph{correct} and \emph{secret}. Correctness is the statement that Alice and Bob share the same key at the end of the protocol, \textit{i.e.}, $K_A=K_B$. Secrecy is the statement that the eavesdropper is totally ignorant about the final key.

\begin{definition}[Correctness]\label{def:correct}
A DIQKD protocol is $\epsilon_{corr}$-correct if the probability that the final key of Alice, $K_A$, differs from the final key of Bob, $K_B$, is smaller than $\epsilon_{corr}$, \textit{i.e.}
\begin{eqnarray}
P(K_A\neq K_B)\leq \epsilon_{corr}.
\end{eqnarray}
\end{definition}

\begin{definition}[Secrecy]\label{def:secret}
Let $\Omega$ denote the event of not aborting in a DIQKD protocol and $p(\Omega)$ be the probability of the event $\Omega$. The protocol is $\epsilon_{sec}$-secret if, for every initial state $\rho_{ABE}$ it holds that
\begin{eqnarray}\label{eq:secret}
p(\Omega)\cdot \frac{1}{2}{\|{\rho_{K_AE}}_{|\Omega}-{\tau_{K_A}\otimes \rho_E}\|}_1\leq \epsilon_{sec},
\end{eqnarray}
where $\tau_{K_A}=\frac{1}{|K_A|}\sum_{k}\ketbra{k}{k}_A$ is the maximally mixed state in the space of strings $K_A$, and ${\|\cdot \|}_1$ is the trace norm.
\end{definition}

If a protocol is  $\epsilon_{corr}$-correct and $\epsilon_{sec}$-secret, then it is $\epsilon^s_{DIQKD}$-correct-and-secret for any $\epsilon^s_{DIQKD}\geq \epsilon_{corr}+\epsilon_{sec}$. See Section \ref{sec:security} for a more detailed definition of security of a DIQKD protocol.

	Given an DIQKD protocol that has $n$ rounds and generates a final correct-and-secret key of $l$ bits, then the secret key rate is defined as
	\begin{equation}
	r=\frac{l}{n}.
	\end{equation}
Our goal is to derive the secret key rate as a function of the parameters of interest, $\beta$ and $Q$, that Alice and Bob can estimate during the execution of the protocol.

\subsection{Security proof of DIQKD}

Even though the BB84 quantum key distribution scheme dates back to 1984 \cite{BB84}, the formal security proof in the asymptotic regime only came out more than a decade later, see e.g.  \cite{May96,May00,LoChau99,SP00}. Security in the composable paradigm in the finite regime against general coherent attacks was only formalized in 2005 \cite{RK05,KR05, RennerThesis}. Moreover, a finite key analysis without the IID assumption over the state preparation and with parameters compatible with current technology only came in 2012 \cite{TLGR12,HT12}.

In the device-independent scenario, security against a quantum eavesdropper\footnote{A discussion on earlier security proofs that do not restrict the eavesdropper to the quantum formalism can be found in \cite{NLreview}.} restricted to collective attacks was first proved in \cite{PAB07, PAB09}. A proof against general attacks assuming memoryless devices was presented in \cite{MPA11,HR10}.  
The problem of extending the security proofs to coherent attacks in the device-independent scenario remained open for a long time. One of the main difficulties is that de Finetti techniques \cite{KR05, CKR09}, used to extend security proofs against collective attacks to general coherent attacks in standard QKD, are not applicable in the DI scenario.
A series of recent works \cite{BCK12b,RUV12,VV14,MS14} culminated in the \emph{Entropy Accumulation Theorem} (EAT) \cite{EATpublish} (see \cite{EAT,DIEAT} for extended versions). The EAT allows one to extend the analysis against collective attacks to the fully device-independent scenario, resulting in asymptotically tight security proofs and high rates in the finite size regime.

\subsection{Experimental DIQKD}\label{sec:introexp}
Protocols for DIQKD rely on a Bell test between two distant parties \cite{PAB09}. In order to certify security, this Bell test should be free of loopholes that could be exploited by an adversary.
While closing  the detection loophole is crucial for a DIQKD implementation, 
the spacelike separation required for loophole-free Bell tests can be relaxed.
In a DIQKD experiment, no-communication between the devices does not have to be guaranteed by spacelike separation, since the assumption of isolated labs, Assumption~\ref{assump:DI}\textit{(\ref{assump:labs})}, is already needed to ensure that the generated key is not leaked to the eavesdropper at any point in time. 
We are thus interested in considering Bell violations between distant - albeit not necessarily spacelike separated - setups in which the detection-loophole is closed \cite{BellDelft,BellVienna,BellNIST,Rosenfeld2017,Giustina2013,Christensen2013,Matsukevich2008,Pironio2010}. The recent performance of fully loophole-free Bell tests \cite{BellDelft,BellVienna,BellNIST,Rosenfeld2017} mark technological progress towards Bell tests without detection loophole over increasingly distant setups, as needed for practically useful DIQKD.

Despite the experimental progress, a device-independent quantum key distribution protocol has not yet been performed. The reason for this is that a Bell violation alone is not enough to guarantee security in a DIQKD protocol. One also needs to account for the amount of information leaked during the error correction, when Alice and Bob correct their string of bits in order to achieve a perfectly correlated raw key. The amount of information required for error correction is determined by the QBER. With a finite QBER, as in practical systems, a large Bell violation is needed to achieve a positive key rate. Moreover, a high minimal number of rounds is required for security due to finite-size effects. The large number of necessary rounds requires a significantly high entangling rate. Altogether, DIQKD demands a low QBER, high Bell violation and high entangling rates. Even though some systems satisfy parts of these requirements, e.g. a high Bell violation \cite{Matsukevich2008,BellDelft,Rosenfeld2017, Pironio2010} or high entangling rate \cite{Giustina2013,Christensen2013,BellVienna,BellNIST}, so far there are no systems that combine all requirements. In section \ref{sec:experiments} we describe the potential platforms for an experimental implementation of DIQKD in detail.

\section{Results}\label{sec:results}

We now present our results. In Section~\ref{sec:rates}, we establish the key rates for DIQKD protocols based on the CHSH inequality, both for coherent and collective attacks in the finite size regime. As a benchmark, in Section \ref{sec:depmodel}, we compare   the key rates that can be achieved in the finite regime for the two adversarial scenarios (collective and coherent attacks) using an implementation with depolarizing noise. In Section \ref{sec:experiments}, we discuss the state of the art of experimental implementations.
	We estimate the parameters of interest for previously performed Bell experiments and situate them in the security proofs.
Additionally, focusing on Nitrogen-vacancy based systems we indicate experimental improvements that can lead to an implementation of DIQKD in the near future.
Throughout this manuscript we use $\mbox{Log}_{10}$ to denote logarithm to base 10 and $\log$ to denote logarithm to base 2.

\subsection{Key Rates}\label{sec:rates}

In the following, we derive the key rates in the finite size regime for DIQKD protocols where the CHSH inequality is used for certifying security. 
For coherent attacks we sharpen the results recently derived in \cite{EATpublish}. For collective attacks we perform the analysis by employing two techniques: the finite version of the asymptotic equipartition property \cite{TCR09} and the additivity of the $2$-R\'enyi entropy.

\subsubsection{Key rates for coherent attacks.}

In order to analyze the key rates against general coherent attacks we use the recently developed entropy accumulation theorem (EAT) \cite{EAT,DIEAT,EATpublish} and consider the following protocol.

\begin{algorithm}[H]
	\floatname{algorithm}{Protocol}
	\caption{DIQKD Protocol for coherent attacks \cite{DIEAT}}\label{prot:diqkd}
	\begin{algorithmic}[1]
		 \For { For every block $j \in [m]$} 
		\State Set $i=0$ and $C_j=\bot$.
		\While {$i \leq s_{max}$} 
		\State Set $i=i+1$.
		\State Alice and Bob choose a random bit $T_i \in \DE{0,1}$ such that $P(T_i=1)=\gamma$.
		\If {$T_i=0$} Alice and Bob choose inputs $(X_i, Y_i)=(0,2)$. 
		\Else { they choose  $X_i ,Y_i \in \DE{0,1}$  (the observables for the CHSH test).}
		\EndIf
		\State Alice and Bob use their devices with the respective inputs and record their outputs, $A_i$ and $B_i$ respectively.
		\State If $T_i=1$ they  set $i=s_{max}+1$.
		\EndWhile
			\EndFor
	\State \textbf{Error Correction:} Alice and Bob apply the error correction protocol $EC$, communicating script $O_{EC}$ in the process. If $EC$ aborts they
		abort the protocol, else they obtain raw keys $\tilde{A}_1^n$ and $\tilde{B}_1^n$.
		\State \textbf{Parameter estimation:} Using $B_1^n$ and $\tilde{B}_1^n$, Bob sets 
		\begin{eqnarray}
		C_i= \left\{\begin{array}{@{}l@{\quad}l}
		1, & \mbox{if  $T_i=1$  and $A_i\oplus B_i=X_i\cdot Y_i$} \\[\jot]
		0, & \mbox{if $T_i=1$  and $A_i\oplus B_i\neq X_i\cdot Y_i$}\\[\jot]
		\bot, & \mbox{if $T_i=0$}
		\end{array}\right.
		\end{eqnarray}
		
		He aborts if 
		\begin{eqnarray*}
			\sum_j C_j<m\times \de{\omega_{exp}-\delta_{est}}(1-(1-\gamma)^{s_{\max}}),
		\end{eqnarray*}
		\ie, if they do not achieve the expected violation. 
		\State \textbf{Privacy Amplification:} Alice and Bob apply the privacy amplification protocol $PA$ and obtain the final keys $K_A$ and $K_B$ of length $l$.
	\end{algorithmic}
\end{algorithm}

In Protocol \ref{prot:diqkd}, the total number of rounds is not fixed in advance, however for a number of blocks $m$ large enough the number of rounds will correspond, with high probability, to the expected value $n$. This is a technicality introduced in Ref.~\cite{DIEAT,EATpublish} in order to obtain better rates in the finite regime. A more detailed explanation can be found in \cite[Appendix B]{DIEAT}. Improvements on the second order term of the entropy accumulation theorem, that do not rely on the introduction of blocks, were recently obtained in \cite{ImproveEAT}.
Following the techniques of \cite{DIEAT,EATpublish}, we derive Theorem \ref{thm:rateEAT}.

\begin{theorem}[Key rates for coherent attacks]\label{thm:rateEAT}
	Either Protocol \ref{prot:diqkd} aborts with probability higher than $1-(\ep_{EA}+\ep_{EC})$, or it generates a
	$(2\ep_{EC}+\ep_{PA}+\ep_s)$-correct-and-secret key  of length
	\begin{eqnarray}
	l\geq& \frac{{n}}{\bar{s}}\eta_{opt} -\frac{{n}}{\bar{s}}h(\omega_{exp}-\delta_{est}) -\sqrt{\frac{{n}}{\bar{s}}}\nu_1  -\mbox{leak}_{EC} \\
	& \quad -3\log\de{1-\sqrt{1-\de{\frac{\epsilon_s}{4(\epsilon_{EA} + \epsilon_{EC})}}^2}}+2\log\de{\frac{1}{2\epsilon_{PA}}},\nonumber
	\end{eqnarray}
	where $\mbox{leak}_{EC}$ is the leakage due to error correction step and the functions $\bar{s}$, $\eta_{opt}$, $\nu_1$ and $\nu_2$ are specified in Table \ref{tab:parametersThmEAT}.
\end{theorem}

Theorem \ref{thm:rateEAT} sharpens the original analysis \cite{DIEAT,EATpublish} and has slightly improved key rates in the finite regime. This results in a reduction of  the minimum  number of rounds (signals) required for positive rates by about a factor of two, as illustrated in Figure~\ref{fig:improveEAT}. A detailed derivation of Theorem~\ref{thm:rateEAT} can be found in \ref{Appendix:rEAT}. 

\begin{figure}[H]
	\caption{Key rate $r$ vs logarithm of the number of rounds $n$. Comparison of the improvements in the key rate, for an implementation where the maximally entangled state is subjected to depolarizing noise and therefore $\beta=2\sqrt{2}(1-2Q)$, for QBER $Q=\DE{0.5\%,2.5\%,5\%}$. The dashed curves correspond to the key rates derived in the original analysis \cite{DIEAT,EATpublish}, the solid lines represent the key rates derived in Theorem~\ref{thm:rateEAT}. Similarly to \cite{EATpublish}, we take $\epsilon^c_{DIQKD}=10 ^{-2}$ and $\epsilon^s_{DIQKD}=10^{-5}$.}\label{fig:improveEAT}
	\centering
	\includegraphics[scale=0.8]{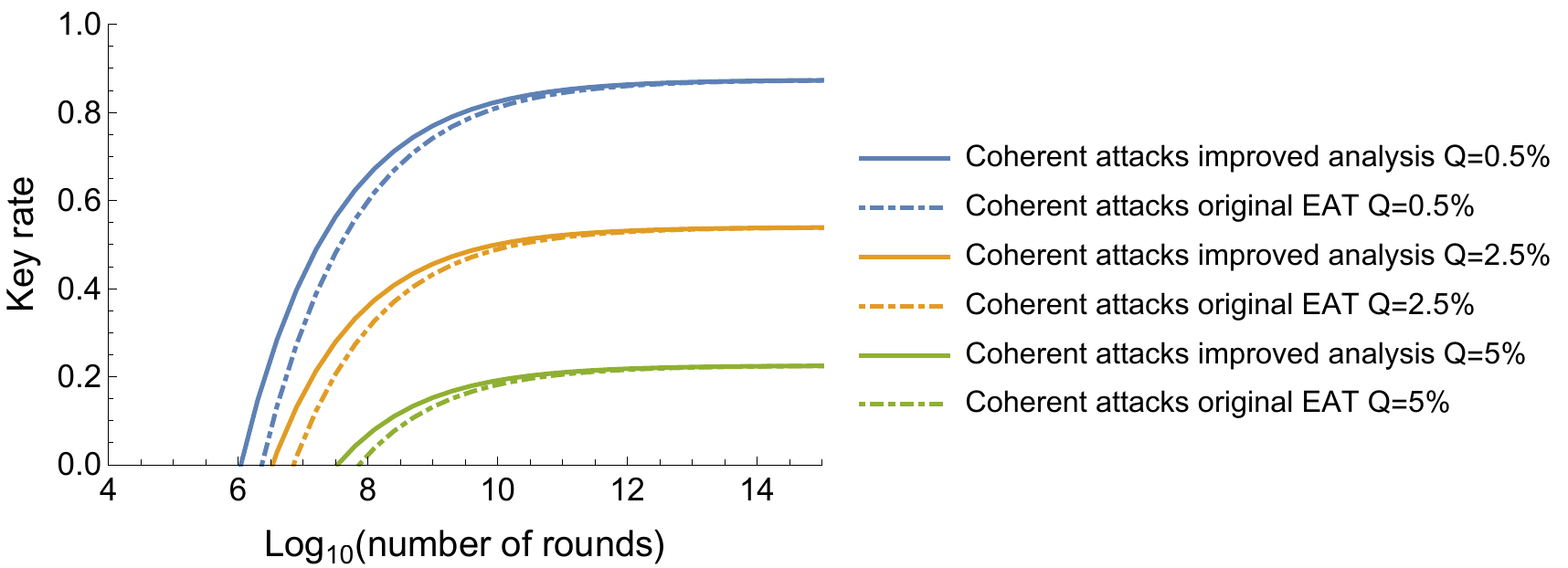}
\end{figure}

\begin{table}[H]
	\centering
	\begin{tabular}{|l|} 
		\hline
		$
s_{\max}=\left\lceil \frac{1}{\gamma} \right\rceil
$\\ \hline
$
\bar{s}=\frac{1-(1-\gamma)^{\left\lceil \frac{1}{\gamma} \right\rceil}}{\gamma}
$\\ \hline
		$\eta_{opt}=\max_{\frac{3}{4}<\frac{{p}_t(1)}{1-(1-\gamma)^{s_{max}}}<\frac{2+\sqrt{2}}{4}} \De{F_{\min}(\vec{p},\vec{p}_t)-\frac{1}{\sqrt{m}}\nu_2}$\\ \hline
			$F_{\min}(\vec{p},\vec{p}_t)=
		\frac{d}{d {p}(1)}g(\vec{p}) \Big|_{\vec{p}_t}\cdot {p}(1)+\de{ g(\vec{p}_t)- \frac{d}{d{p}(1)}g(\vec{p})\Big|_{\vec{p}_t}\cdot {p}_t(1) }$  \\  \hline	
			$ g({\vec{p}}) =  
		{s}\De{1-h\de{\frac{1}{2}+\frac{1}{2}\sqrt{16\frac{{p}(1)}{1-(1-\gamma)^{s_{max}}}\de{\frac{{p}(1)}{1-(1-\gamma)^{s_{max}}} -1}+3 } }}$\\\hline	
		$\nu_2 =2 \de{\log\de{1+2\cdot 2^{s_{\max}}3}+\left\lceil \frac{d}{d{p}(1)}g(\vec{p})\big|_{\vec{p}_t}\right\rceil}\sqrt{1-2\log \epsilon_s}$\\ \hline
$	\nu_1=2 \de{\log 7 +\left\lceil\frac{|h'(\omega_{exp}+\delta_{est})|}{1-(1-\gamma)^{s_{\max}}}\right\rceil}\sqrt{1-2\log\epsilon_s}
$\\ \hline
	\end{tabular}
	\caption{Explicit form of the terms that appear in Theorem~\ref{thm:rateEAT}. For a detailed derivation see \ref{Appendix:rEAT}.}
	\label{tab:parametersThmEAT}
\end{table}

\subsubsection{Key rates for collective attacks}

For collective attacks, we derive the finite key rates by employing two techniques: the finite version of the asymptotic equipartition property and the additivity property of the conditional $\alpha$-R\'enyi entropies.
To deal with collective attacks we can use a simplified version of Protocol~\ref{prot:diqkd}, where the number of rounds is fixed. 

\begin{algorithm}[H]
	\floatname{algorithm}{Protocol}
	\caption{DIQKD protocol for collective attacks}\label{prot:diqkdIID}
	\begin{algorithmic}[1]
		  \For {$i=1$ to $n$} 
		  \State Alice and Bob choose a random bit $T_i \in \DE{0,1}$ such that $P(T_i=1)=\gamma$.
		\If {$T_i=0$} Alice and Bob choose inputs $(X_i, Y_i)=(0,2)$. 
		\Else { they choose  $X_i ,Y_i \in \DE{0,1}$  (the observables for the CHSH test).}
		\EndIf
		\State Alice and Bob use their devices with the respective inputs and record the outputs, $A_i$ and $B_i$ respectively.
		   \EndFor
		\State \textbf{Error correction:} Alice and Bob apply the error correction protocol $EC$, communicating $O_{EC}$ in the process. 
			If $EC$ aborts they
		abort the protocol, else they obtain raw keys $\tilde{A}_1^n$ and $\tilde{B}_1^n$.
		\State \textbf{Parameter estimation:} Using $B_1^n$ and $\tilde{B}_1^n$, Bob sets for the first test rounds
		\begin{eqnarray}
		C_i= \left\{\begin{array}{@{}l@{\quad}l}
		1, & \mbox{if $A_i\oplus B_i=X_i\cdot Y_i$} \\[\jot]
		0, & \mbox{if $A_i\oplus B_i\neq X_i\cdot Y_i$}\\[\jot]
		\end{array}\right.
		\end{eqnarray}
		For the remaining rounds he sets $C_i=\bot$.
		
		He aborts if 
		\begin{eqnarray*}
			\sum_j C_j<\gamma n\times \de{\omega_{exp}-\delta_{est}},
		\end{eqnarray*}
		\ie, if they do not achieve the expected violation. 
		\State \textbf{Privacy Amplification:} Alice and Bob apply the privacy amplification protocol $PA$ and obtain the final keys $K_A$ and $K_B$ of length $l$.
	\end{algorithmic}
\end{algorithm}

In the following theorem we state the length of a secure key that can be derived using the asymptotic equipartition property, which is formally stated in Theorem \ref{thm:AEP}.

\begin{theorem}\label{thm:rateIID} 
	Either Protocol \ref{prot:diqkdIID} aborts with probability higher than $1-(\epsilon_{con}+\epsilon_{EC})$, or it generates a $(2\epsilon_{EC}+\epsilon_s+\epsilon_{PA})$-correct-and-secret  key of length:
	\begin{eqnarray}
	l\geq n& \big[ 1-h\de{\frac{1}{2}+\frac{1}{2}\sqrt{16 (\omega_{exp}-\delta_{est}-\delta_{con})((\omega_{exp}-\delta_{est}-\delta_{con})-1)+3}}\nonumber\\
	& -(1-\gamma)h(Q)-\gamma h(\omega_{exp})\big]\\
	&\;\;\; -\sqrt{n}\de{4\log\de{2\sqrt{2}+1} \de{\sqrt{\log\frac{2}{\epsilon_s^2}}+ \sqrt{\log \frac{8}{{\epsilon'}_{EC}^2}}}}\nonumber\\
	&\;\;\; -\log\de{\frac{8}{{\epsilon'}_{EC}^2}+\frac{2}{2-\epsilon'_{EC}} }-\log \de{\frac{1}{\epsilon_{EC}}}- 2\log\de{\frac{1}{2\epsilon_{PA}}}\nonumber
	\end{eqnarray}
\end{theorem}

A detailed derivation of Theorem \ref{thm:rateIID} can be found in \ref{Appendix:rIID}.\vspace{1em}

Using a different technique, namely bounding the key rate by the conditional collision entropy, we derive the following result.

\begin{theorem}\label{thm:rateH2} 
	Either Protocol \ref{prot:diqkdIID} aborts with probability higher than $1-(\epsilon_{con}+\epsilon_{EC})$, or it generates a $(2\epsilon_{EC}+\epsilon_{PA})$-correct-and-secret  key of length:
	\begin{eqnarray}
	l\geq n& \Big[-\log\de{\frac{1}{2}+\frac{1}{2}\sqrt{16(\omega_{exp}-\delta_{est}-\delta_{con})(1-(\omega_{exp}-\delta_{est}-\delta_{con}))-2}}	\nonumber\\
	& -(1-\gamma)h(Q)-\gamma h(\omega_{exp})\Big]\nonumber\\
	&\;\;\; -\sqrt{n}\de{4\log\de{2\sqrt{2}+1}  \sqrt{\log \frac{8}{{\epsilon'}_{EC}^2}}}\\
	&\;\;\; -\log\de{\frac{8}{{\epsilon'}_{EC}^2}+\frac{2}{2-\epsilon'_{EC}} }\nonumber \\
&	-\log \de{\frac{1}{\epsilon_{EC}}}- 2\log\de{\frac{1}{2\epsilon_{PA}}}\gm{-2\log\de{\frac{1}{\epsilon_{con}+\epsilon_{EC}}}}.\nonumber
	\end{eqnarray}
\end{theorem}

An important step in the proof of Theorem~\ref{thm:rateH2} is to derive a lower bound on the  collision entropy as a function of the CHSH violation $\beta$. A tight lower bound is proved in Theorem~\ref{thm:H2}. The detailed proof of Theorem~\ref{thm:rateH2}  is presented in \ref{Appendix:rH2}.

The rates presented in Theorem~\ref{thm:rateIID} are asymptotically tight, while Theorem~\ref{thm:rateH2} achieves strictly smaller asymptotic rates.
However, one can note that in Theorem~\ref{thm:rateH2} the term proportional to $\sqrt{n}$ has a smaller pre-factor. This can potentially lead to an advantage for the minimum number of rounds required for security. For Protocol~\ref{prot:diqkdIID}, an advantage can only be observed for very low noise regime, as illustrated in Figure \ref{fig:rateHxH2}. We remark, however, that for protocols based on other Bell inequalities the techniques used for deriving Theorem~\ref{thm:rateH2} can present significant advantage for the collective attack analysis.  This is further discussed in Section \ref{sec:Hevaluation}.

 \begin{figure}[H]
 	\caption{Key rates vs logarithm of the number of rounds $n$ for Protocol~\ref{prot:diqkdIID} (collective attacks). The blue curve represent the key rate using Theorem~\ref{thm:rateIID} and the yellow curve shows the key rate using Theorem~\ref{thm:rateH2}. It is considered an implementation with depolarizing noise and QBER $Q=0.01\%$. The inset graph shows a zoom in the region of low number of rounds. Similarly to \cite{EATpublish}, we take $\epsilon^c_{DIQKD}=10 ^{-2}$ and $\epsilon^s_{DIQKD}=10^{-5}$.}\label{fig:rateHxH2}
 	\centering
 	\includegraphics[scale=0.8]{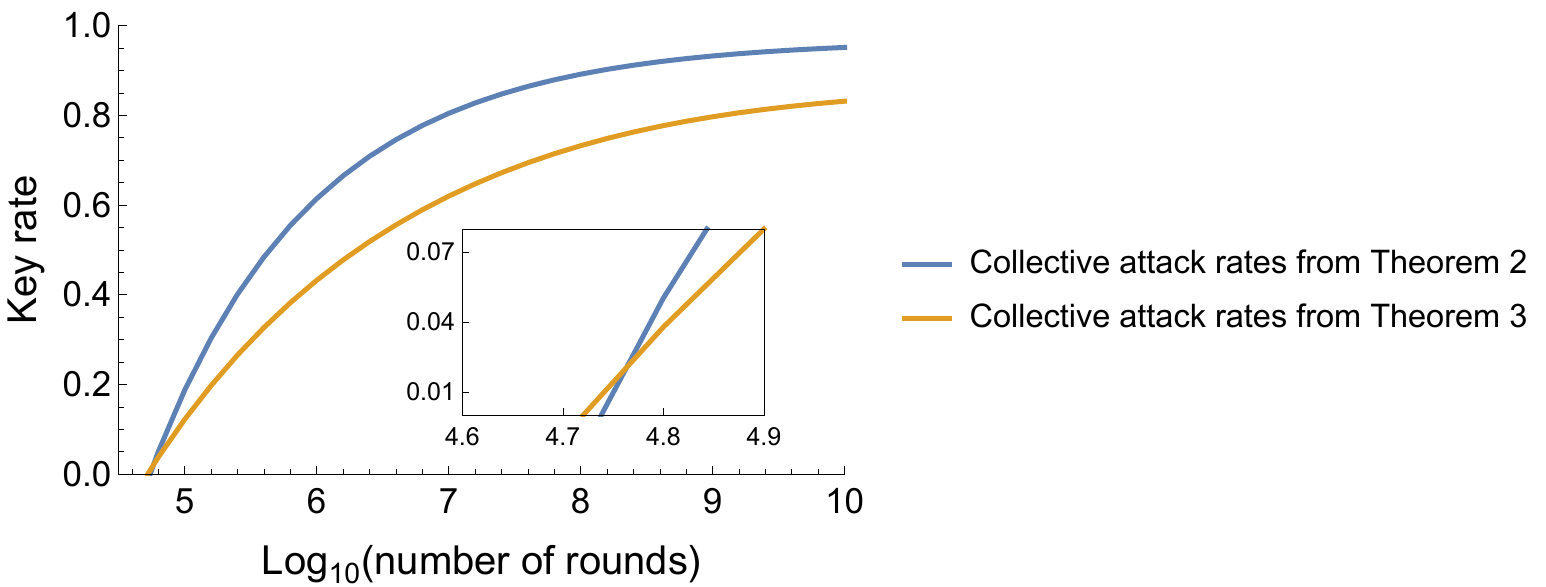}
 \end{figure}

The following table lists the parameters of the DIQKD protocols in consideration.
 \begin{table}[H]
	\centering
	\begin{tabular}{|c|l|} 
		\hline
		$n$& expected number of rounds \\ \hline
		$l$& final key length \\\hline
		$\gamma$& fraction of test rounds \\ \hline
		$Q$& quantum bit error rate\\\hline
		$\beta$& CHSH violation \\\hline
		$\omega_{exp}$ & expected winning probability on the CHSH game in an honest implementation \\ \hline
		$\delta_{est}$& width of the statistical interval for the Bell test\\\hline
		$\delta_{con}$&confidence interval for the Bell test in Protocol \ref{prot:diqkdIID} \\ \hline
$\epsilon_s$& smoothing parameter \\ \hline
$\epsilon_{EC},\epsilon'_{EC}$& error probabilities of the error correction protocol \\ \hline
$\epsilon_{EA}$& error probability of Bell violation estimation in Protocol \ref{prot:diqkd} \\\hline
$\epsilon_{con}$& error probability of Bell violation estimation in Protocol \ref{prot:diqkdIID} \\\hline
$\epsilon_{PA}$& error probability of the privacy amplification protocol \\ \hline
$\mbox{leak}_{EC}$&leakage in the error correction protocol \\\hline
	\end{tabular}
	\caption{Parameters of the considered DIQKD protocols, Protocol \ref{prot:diqkd} and Protocol \ref{prot:diqkdIID}.}
	\label{tab:parametersKey}
\end{table}

\subsection{Comparison of key rates for depolarizing noise model}\label{sec:depmodel}

We now compare the key rates achieved in the finite regime under the assumption of collective attacks (IID scenario) and against general coherent attacks (fully DI scenario). As a benchmark, we focus on an honest implementation where the maximally entangled state is prepared and subjected to depolarizing noise\footnote{This noise model can also be seen as the case where each individual qubit suffers a depolarization with parameter $\nu'$, where $\nu=2\nu'-{\nu'}^2$.}: 
\begin{equation}\label{eq:depstate}
\rho=(1-\nu)\ketbra{\Phi^+}{\Phi^+}+\nu\frac{I}{4}.
\end{equation}
In this case, the parameters of interest -- the value of the CHSH inequality $\beta$ and the QBER $Q$ -- relate to the noise parameter 
$\nu$ by 
\begin{eqnarray}\label{eq:dep}
Q=\frac{\nu}{2}\; \mbox{ and }\; \beta=2\sqrt{2}(1-\nu) \; \rightarrow \; \beta=2\sqrt{2}(1-2Q).
\end{eqnarray}

In Figure~\ref{fig:CompareCohxIID} we compare the key rates achievable under the IID assumption, given by Theorem~\ref{thm:rateEAT}, and in the fully DI scenario, Theorem~\ref{thm:rateIID}, for an honest implementation with depolarizing noise. 

\begin{figure}[H]
  \caption{Key rates vs logarithm of the number of rounds for collective attacks (dashed lines) and coherent attacks (solid lines). The different curves represent different values of QBER $Q=\de{0.5\%,2.5\%.5\%}$  considering an implementation where the maximally entangled state is subjected to depolarizing noise (see relation (\ref{eq:dep})).  The security parameters are taken as $\epsilon^c_{DIQKD}=10 ^{-2}$ and $\epsilon^s_{DIQKD}=10^{-5}$.}\label{fig:CompareCohxIID}
  \centering
    \includegraphics[scale=0.8]{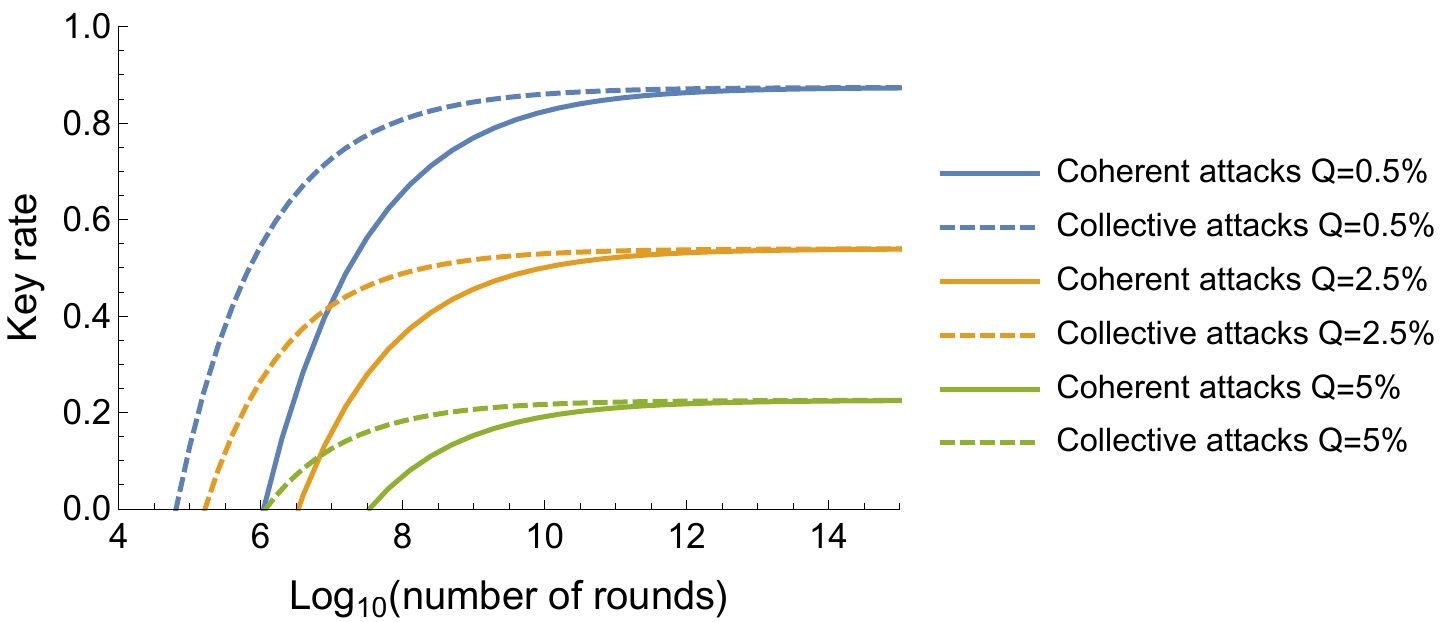}
\end{figure}

Figure~\ref{fig:CompareCohxIID} shows that the key rates approach the same asymptotic values, however the minimum number of rounds required to guarantee security is significantly higher for general coherent attacks. 
Indeed, by adding the assumption that the eavesdropper is restricted to collective attacks, the minimum number of signals required to have a positive key rate drops by about two orders of magnitude. 
However, even for collective attacks, this minimum number of required rounds is considerably large given the current entanglement generation rates. This is one of the big challenges to be overcome for a DIQKD implementation. 
In the next Section we are going to discuss the state of the art of experiments, and situate the current achievable parameters (Bell violation, QBER and entanglement generation rate) in the security proofs.

\subsection{The state-of-the-art experimental DIQKD}\label{sec:experiments}

In the following, we discuss experimental platforms in which DIQKD may be implemented. We analyse Bell violations and expected QBER achieved in previous Bell tests with distant setups and situate these parameters in the context of the key rates derived in Theorems \ref{thm:rateEAT} and \ref{thm:rateIID}. A summary of the findings is presented in Table \ref{tab:Bellparameter} and Figures \ref{fig:Ncoh} and \ref{fig:Niid}.

In experimental setups, distant entanglement is typically generated using photons to establish the connection. We distinguish two approaches based on the role of the photonic qubits:  \textit{(i) All-photonic schemes:} Approaches in which the entangled state is encoded in the photonic state directly. In this case, measurements of the photonic states on two remote setups enable to infer  their entanglement. 
\textit{(ii) Heralded schemes:} In this case, the entangled state is typically created in a long-lived system and the photons are used as a means of establishing the entanglement between two distant systems.

In this section we provide a discussion of the parameters in each of these schemes and the related challenges towards an implementation of DIQKD. We provide a more detailed discussion of one of the systems, namely nitrogen-vacancy (NV) centres in diamonds, and describe improvements in experimental parameters that can lead to a DIQKD implementation in the near future.

\subsubsection{DIQKD with all-photonic entanglement.}
Since in all-photonic schemes the entangled state is directly encoded on the photonic state, photon losses limit the entangled state detection efficiency. Closing the detection loophole in a Bell test thus requires very efficient entangled-photon sources and photon detectors. Recent technological advances enabled all-photonic Bell tests that close the detection-loophole \cite{Christensen2013,Giustina2013}, later combined with spacelike separation in loophole-free Bell tests \cite{BellVienna,BellNIST}. 

In photonic systems the detection efficiency also impacts the entangled state fidelity. We thus may expect that Bell violations are low in photonic systems. To avoid having to deal with undetected events, photonic Bell tests typically employ the CH-Eberhard inequality \cite{Clauser1974,Eberhard1993}. 
 The CHSH and CH-Eberhard inequalities are equivalent\footnote{One can see that by replacing non-detected events by the deterministic classical strategy ``output 1" in a test of the CHSH inequality.}, such that we can estimate the CHSH violation achieved in photonic experiments.
 Table \ref{tab:Bellparameter} presents the corresponding value for the CHSH inequality achieved in the experiments of Refs. \cite{Christensen2013,Giustina2013,BellVienna,BellNIST}. One can note that the violations achieved are indeed low, ranging from $2.00004$ to $2.02$. Combined with a finite QBER ($>2\%$), this poses a significant challenge for the implementation of a DIQKD protocol in photonic systems.

However, if these systems would enter the regime of positive key rates, the entanglement generation rate can be very high ($\sim 10^5$ Hz), such that they could easily reach the asymptotic key rate values. 

\gm{In order to overcome photon losses, several proposals for implementing heralding schemes in all-photonic systems were presented. In this case, the entangled state is created between photons and, also,  this entanglement is heralded by the interference of other photons. In particular, in Ref.~\cite{Gisin10} the authors propose a scheme based on a qubit amplifier that combines single photon sources and linear optics. This proposal was further explored in Ref.~\cite{PMWLL11}. Schemes based on entanglement swapping by quantum relay were also considered \cite{CT11,Meyer13,STS16}. Ref.~\cite{CT11} makes a comparison of the performance of the two types of schemes. Analyses in Refs. \cite{Gisin10,CT11,Meyer13,STS16} make assumptions on the possible attacks performed by the eavesdropper.
	New protocols based on single photon sources were recently proposed in Ref.~\cite{Alejandro}. The proposed schemes uses a combination of spontaneous parametric down conversion sources and single-photon sources in order to achieve a setup where a heralding process could overcome transmission photon losses. The security analysis presented in Ref.~\cite{Alejandro} does not restrict the eavesdropper attacks. These setups are a promising proposal to bring the parameters of all-photonic systems to the region of positive asymptotic key rates (see Figure~\ref{fig:Ncoh} and \ref{fig:Niid}). However single-photon sources still lack the required performance for an implementation of these schemes.	
}


\subsubsection{DIQKD with heralded entanglement.}
Due to the nature of heralded entangling schemes, photon losses do not influence the entangled state detection efficiency or fidelity. 
Heralded schemes have been used to entangle distant atomic ensembles \cite{Chou2005, Usmani2012}, trapped ions \cite{Moehring2007}, atoms \cite{Hofmann2012}, NV centres \cite{Bernien2013}, quantum dots \cite{Delteil2015}, and mechanical oscillators \cite{Riedinger2017}. So far, entangled state fidelities sufficiently high to violate Bell's inequalities have only been reached with trapped ions \cite{Matsukevich2008,Pironio2010}, atoms \cite{Hofmann2012,Rosenfeld2017}, and with NV centres \cite{BellDelft,Bellsecondrun}. The Bell violations observed in Refs. \cite{Matsukevich2008,Pironio2010,Rosenfeld2017,BellDelft,Bellsecondrun}  are in the range $\beta = 2.22$ to $\beta=2.41$, with a lower bound on the QBER, estimated from detection efficiencies alone, around $0.04$ (see Table \ref{tab:Bellparameter} for a full overview). Apart from the results reported in \cite{Pironio2010}, these parameters are not in the region of positive key rate (see Figures \ref{fig:Ncoh} and \ref{fig:Niid}). However, all of them 
are in the proximity of this region, such that setup improvements may enable to reach it. 

The challenge for these implementations is however their low entangling rate, induced by photon losses. Current rates range from (minutes)$^{-1}$ \cite{Matsukevich2008,Hofmann2012,Pironio2010,Rosenfeld2017} to (hours)$^{-1}$ \cite{BellDelft,Bellsecondrun}. 
A significant speed-up in the entanglement generation rate is thus needed in order to achieve the minimum number of rounds required for DIQKD. Higher entangling rates in heralded schemes were recently achieved with trapped ions \cite{Hucul15} and NV centres \cite{Kalb2017,Humphreys2018}, although with lower state fidelities, and no Bell violations are reported. Even though in Ref. \cite{Humphreys2018} the state fidelity is just high enough to be able to violate Bell inequalities, the expected Bell violation would be low. Enhancement in entangling rates, e.g. with optical cavities to improve light-matter coupling efficiency \cite{Reiserer2014} is therefore crucial to achieving an implementation of DIQKD with heralded schemes.

\subsubsection{Nitrogen-vacancy centre-based networks.}\label{sec:expNV}
In this section, we focus on heralded entanglement generation between nitrogen-vacancy centres in diamond for DIQKD.
Nitrogen-vacancy (NV) centres are defect centres in the diamond lattice. They contain an electronic spin with good coherence properties and spin-selective optical transitions that can be used for intialization, readout and entanglement generation \cite{Bernien2013,Awschalom2018}. Next to the electronic spin, nearby weakly coupled nuclear spins can serve as long-lived memories \cite{Cramer2015,Kalb2018}. These properties make the NV centre a promising quantum network node. 
 
Entanglement between distant NV centres can be generated using an heralded scheme. Typically, local entanglement is first generated between the NV electronic spin and a photon mode. And subsequently, entanglement between distant NV centres is achieved through entanglement swapping by interfering the two photon modes from distant setups \cite{Barrett2005}. As discussed above for heralded protocols, photon attenuation does not influence the fidelity of the generated entangled state or the detection efficiency. The detection of the spin states has near-unit efficiency \cite{Robledo2011}.

\subsubsection*{DIQKD parameters.}
In a loophole-free Bell test with NV centres \cite{BellDelft,Bellsecondrun}, a CHSH violation $\beta=2.38\pm0.14$ was observed between systems separated by 1.3 kilometers. Taking into account the entangled state fidelity and detection efficiency, we estimate that the corresponding QBER would be $Q= 0.06\pm0.03$. The Bell violation achieved in \cite{BellDelft,Bellsecondrun} is considerably high, especially if compared to loophole-free Bell test experiments in photonic systems \cite{BellVienna,BellNIST}. However, these parameters are not good enough to generate a secure key. Indeed, using  Theorems \ref{thm:rateEAT} and \ref{thm:rateIID}, one concludes that it is not possible to achieve positive key rate with these parameters (see Figures \ref{fig:Ncoh} and \ref{fig:Niid}).

In the following, we suggest two near-term experimental improvements to enhance these parameters.

Firstly, the frequency stability of the laser used to excite NV centres during the entanglement protocols can be increased using an external cavity. The instability of the laser can influence the indistinguishability of photons emitted by the distant NV centres. The indistinguishability is crucial for photon interference, which can be quantified by the visibility of the two-photon quantum interference (TPQI). We expect that compared to previous implementation \cite{BellDelft}, the improved laser frequency stability
can lead to an improvement in TPQI visibility from 0.88 to 0.90. 

Secondly, both the CHSH violation $\beta$ and the QBER $Q$ are impacted by the NV electronic spin state readout. The readout can be performed using resonant excitation of a spin-selective optical transition \cite{Robledo2011}. 
Improvements to the detection efficiency can be obtained by storing the spin state in the nearby nitrogen spin state, and performing repeated readout \cite{Jiang2009}. We estimate that the repeated readout can lead to an average readout fidelity of $\approx0.985$, compared to an initial 0.97 \cite{Pfaff2014} \footnote{We note that this readout method increases the readout duration, which compromises spacelike setup-separation. However, security in a DIQKD implementation does not require spacelike separation since it is superfluous with the assumption of isolated labs in place (see Assumptions \ref{assump:DI}). Therefore, an increased readout time does not present a problem for security.}.

Other improvements can be envisioned, such as enhancement of the detection efficiency by improving the photon collection efficiency through the use of parabolic reflectors \cite{Wan2018} or optical cavities \cite{Faraon2011}. In the following discussion we limit ourselves to the two advances listed above and summarized in Table \ref{tab:standard_params}.

\begin{table}[h]
\centering
\begin{tabular}{l||cc|cc}
 \bf{DIQKD parameters} &  \multicolumn{2}{c}{Ref. \cite{BellDelft, Bellsecondrun}}\vline & \multicolumn{2}{c}{Expected}\\ 
setup  &  A  &  B  &  A  &  B \\ \hline
average readout fidelity & 0.974 & 0.969   & 0.985 & 0.985 \\ 
TPQI visibility & \multicolumn{2}{c}{0.88} \vline &  \multicolumn{2}{c}{0.90}\\  \hline
$\beta$ & \multicolumn{2}{c}{2.38 $\pm$ 0.14}\vline & \multicolumn{2}{c}{2.47}\\ 
Q &  \multicolumn{2}{c}{0.06 $\pm$ 0.03} \vline & \multicolumn{2}{c}{0.051}\\
\end{tabular}
\caption{The CHSH violation $\beta$ and QBER $Q$ in NV centre-based implementations are strongly dependent on the TPQI visibility and the readout fidelity. The resulting values are shown for parameters achieved in a loophole-free Bell test, and for expected values from several readily-implementable improvements.}
\label{tab:standard_params}
\end{table}

Taking into account these improvements, the expected DIQKD parameters are $\beta \approx 2.47$ and $Q \approx 0.051$. 
In Figure~\ref{fig:NV} we illustrate the rates achievable for these parameters against general coherent attacks and under the assumption that the eavesdropper is restricted to collective attacks. We see that the required minimum number of rounds is of order $10^8$ for general attacks, and about $5\times 10^6$ for collective attacks.

\begin{figure}[H]
\begin{center}
	\includegraphics[scale=0.8]{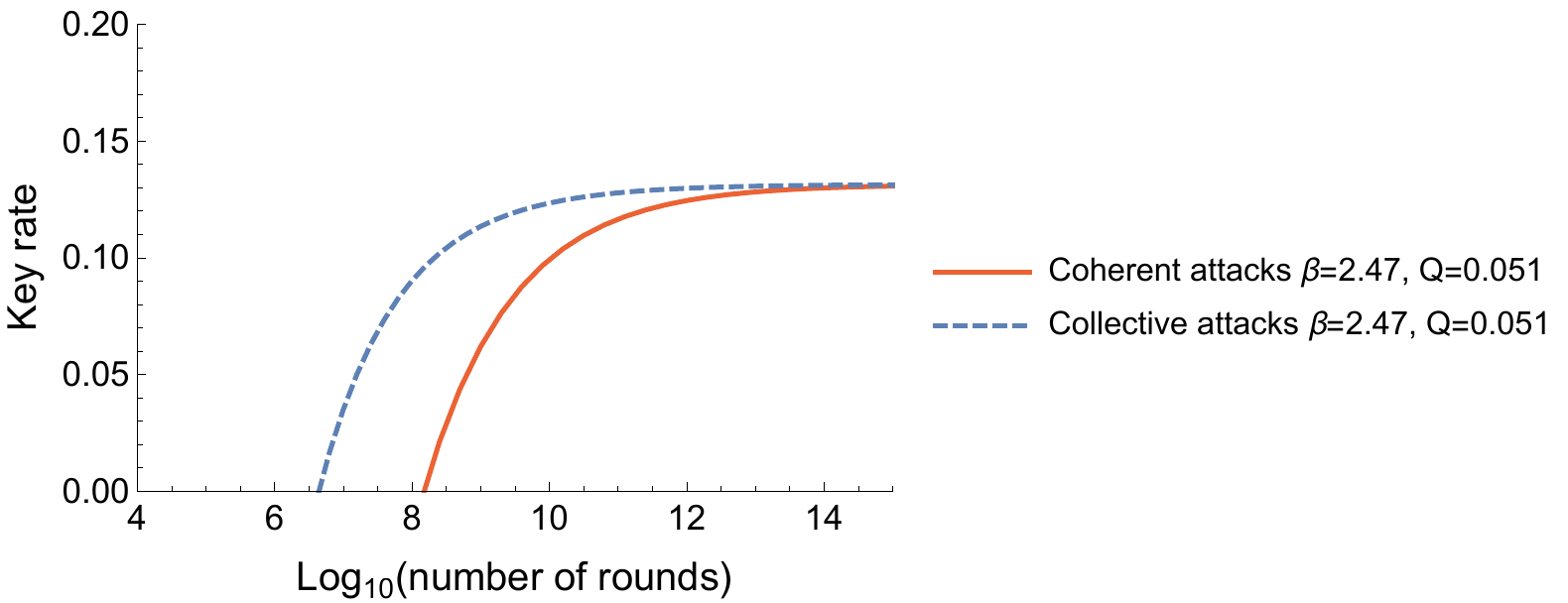}
	\caption{Key rates vs logarithm of the number of rounds $n$ for parameters that are readily-implementable in NV centres setups (CHSH violation $\beta=2.47$ and QBER $Q=0.051$). The red line shows the key rates obtained against general coherent attacks, and the blue dashed line shows the key rates under the assumption of collective attacks. The security parameters are chosen to be $\epsilon^c_{DIQKD}=10 ^{-2}$ and $\epsilon^s_{DIQKD}=10^{-5}$.}\label{fig:NV}
		\end{center}
\end{figure}

\subsubsection*{Entangling rate.}
Although the improved parameters lead to a positive key rate, this does not mean that DIQKD with NV centres is readily achievable. The system faces another challenge: the probabilistic nature of the heralded entanglement scheme limits the entanglement generation rate.

In the heralded entanglement generation protocol used in \cite{Bernien2013,BellDelft} the photonic qubit is time-bin encoded and entanglement is heralded with the detection of a photon in each of two time-bins \cite{Barrett2005}. Since two photons have to be detected, the rate of the protocol is proportional to the square of the photon losses. For the spacelike separated setups in \cite{BellDelft} the total emission and detection efficiency per photon is $\approx 10^{-4}$, leading to a total success probability of $\approx 10^{-8}$. Since the repetition rate, limited by the spin-state reset time, is of the order of $\approx \mu$s, generating a raw key of length $10^6$ bits would take $\approx 10^3$ days. It is clear that a speed-up of entanglement generation rate is required to use NV centres in a DIQKD protocol. We describe two approaches toward this.

Firstly, this could be achieved by adapting the entanglement generation protocol. A linear dependency of the rate on photon losses can be achieved by employing an extreme-photon-loss (EPL) protocol \cite{Campbell2008} or single-photon (SP) protocol \cite{Cabrillo1999}. Demonstrated implementations of these protocols with NV centres indeed provide a speed-up in entanglement rate of three orders of magnitude \cite{Kalb2017, Humphreys2018}. However, these implementations do not yet provide the entangled state fidelities leading to Bell violations that allow for DIQKD (the entangled state fidelities are $F_{EPL} = 0.65\pm 0.03$ and $F_{SP} = 0.81\pm0.02$, leading to no Bell violation for the EPL protocol and a small violation $\beta_{SP} = 2.1$ for the single photon protocol). Better parameters may be achieved with improvements of the robustness of the nuclear-spin memories \cite{Kalb2018} and with an improved photon detection versus dark-count rate \cite{Cabrillo1999}.

Secondly, an increase in the entanglement rate can be achieved by a reduction of the photon losses per round. These losses consist of three parts: a low coherent-photon emission probability, a non-unit collection efficiency and fiber attenuation.
The photon attenuation during transmission over fibers is $\approx 8$ dB for the NV emission wavelength (637 nm).
To maintain high entangling rates for distant setups, this should be reduced. This can be achieved by frequency downconversion of the photons at a wavelength of 637 nm emitted by the NV centres to telecom frequencies \cite{Bock2018,Dreau2018}.
The emission probability of coherent photons, $\approx3\%$, and subsequent collection efficiency ($\approx 10\%$, \cite{Bernien2013}) together limit the best achievable entangling rates. They can be addressed simultaneously by embedding the NV centre in an optical cavity to enhance coherent-photon emission and the collection efficiency  \cite{Faraon2011}. A promising approach employs NV centres in diamond membranes in Fabry-Perot microcavities \cite{Janitz2015,Bogdanovic2017,Riedel2017}. In such a design NV centres remain far away from the optical interface, retaining bulk-like optical coherence properties. These cavities are expected to provide three orders of magnitude enhancement in entangling rate for a two-click protocol \cite{Bogdanovic2017}. Together with the improved DIQKD parameters described above, this makes a demonstration of DIQKD with NV centres experimentally feasible.

\begin{center}
\begin{table}[H]
	\begin{tabular}{ |l|c|c|c| } 
		\hline
		 & $\beta$ & $Q$ \\
		\hline
		(1) Matsukevich et al., PRL 100, 150404 (2008) \cite{Matsukevich2008} & $2.22\pm 0.07$ & $0.041\pm0.003$ \\ \hline
		(2) Pironio et al., Nature 464, 1021-1024 (2010) \cite{Pironio2010}& $2.414\pm 0.058$ & $0.041\pm0.003$ \\ \hline
		(3) Giustina et al., Nature 497,  227-230 (2013) \cite{Giustina2013} &$2.02096\pm 0.00032$ & $0.0297\pm0.0003$\\ \hline
		(4)	Christensen et al., PRL 111, 130406 (2013) \cite{Christensen2013} & $2.00022\pm 0.00003$ & $0.0244 \pm 0.0009$\\ \hline
			
		(5) Giustina et al., PRL 115, 250401 (2015) \cite{BellVienna} & $2.000030\pm 0.000002$ & $0.0379\pm 0.0002$ \\  \hline
		(6) Shalm et al., PRL 115, 250402 (2015) \cite{BellNIST} & $2.00004\pm0.00001$ &  $0.0292\pm 0.0002$ \\ \hline
		
		(7) Hensen et al., Nature 526 682-686 (2015) \cite{BellDelft} &  $2.38\pm 0.14$ & $0.06\pm0.03$\\\hline
		(8) Rosenfeld et al., PRL 119, 010402 (2017) \cite{Rosenfeld2017} & $2.221\pm0.033$ & $0.035\pm0.003$ \\ \hline
		(9) Expected improvements in NV systems & 2.47 &0.051\\
		\hline
	\end{tabular}
	\caption{Summary of the estimated parameters of interest for DIQKD. (1,2) are Bell tests with trapped ions, (3-5) are all-photonic experiments, (7) uses NV centres and (8) trapped atoms. (9) reports on near-term achievable parameters with NV centers as described in Section~\ref{sec:expNV}. In all experiments the detection loophole is closed; (5-8) additionally close the locality loophole.
The CHSH violations for neutral atoms (8), trapped ions (1,2) and NV centres (7)  are as reported in the corresponding experiments. For (3), (4) and (5), in which the value of the CH-Eberhard inequality $J$ is reported, we make use of the relation $\beta=4J + 2$ between the CHSH value and the CH-Eberhard value. This relation is found if one attributes ``output 1" to undetected events in a CHSH inequality test. For (6) the CHSH violation was estimated directly from the reported data.
For the estimation of the QBER ($Q$), in (1),(2) and (8) we assume perfect classical correlation in the generated state and find a lower bound for the QBER from reported detection efficiencies ($0.979\pm0.002$ \cite{Olmschenk07} for (1) and (2), and $0.982\pm0.002$ \cite{Henkel10} for (8)). 
For NV centres (7), we additionally account for imperfections in the entangled state based on the reported density matrix. 
 For all-photonic systems (3-6), the QBER is estimated by taking into account the detection efficiency and using the reported estimated state and the measurements performed by Alice, optimizing over measurements for Bob.
}\label{tab:Bellparameter}
	\end{table}
\end{center}

\begin{figure}[H]
  \caption{Region of positive key rates for \emph{coherent attacks}: The red area is the region of values of QBER ($Q$) and CHSH violation ($\beta$) for which a positive key rate cannot be reached with any number of rounds. In the green area, the dashed curves represents the minimum number of rounds required to get positive key rate. For parameters above each curve, a key rate can be extracted if the number of rounds is higher than specified in the curve. The points show the Bell violation and estimated QBER achieved by previous experiments (see Table \ref{tab:Bellparameter}). They, however, do not reflect the corresponding entanglement generation rates. Similarly to \cite{EATpublish}, we take $\epsilon^c_{DIQKD}=10 ^{-2}$ and $\epsilon^s_{DIQKD}=10^{-5}$.}\label{fig:Ncoh}
  \begin{center}
    \includegraphics[scale=0.65]{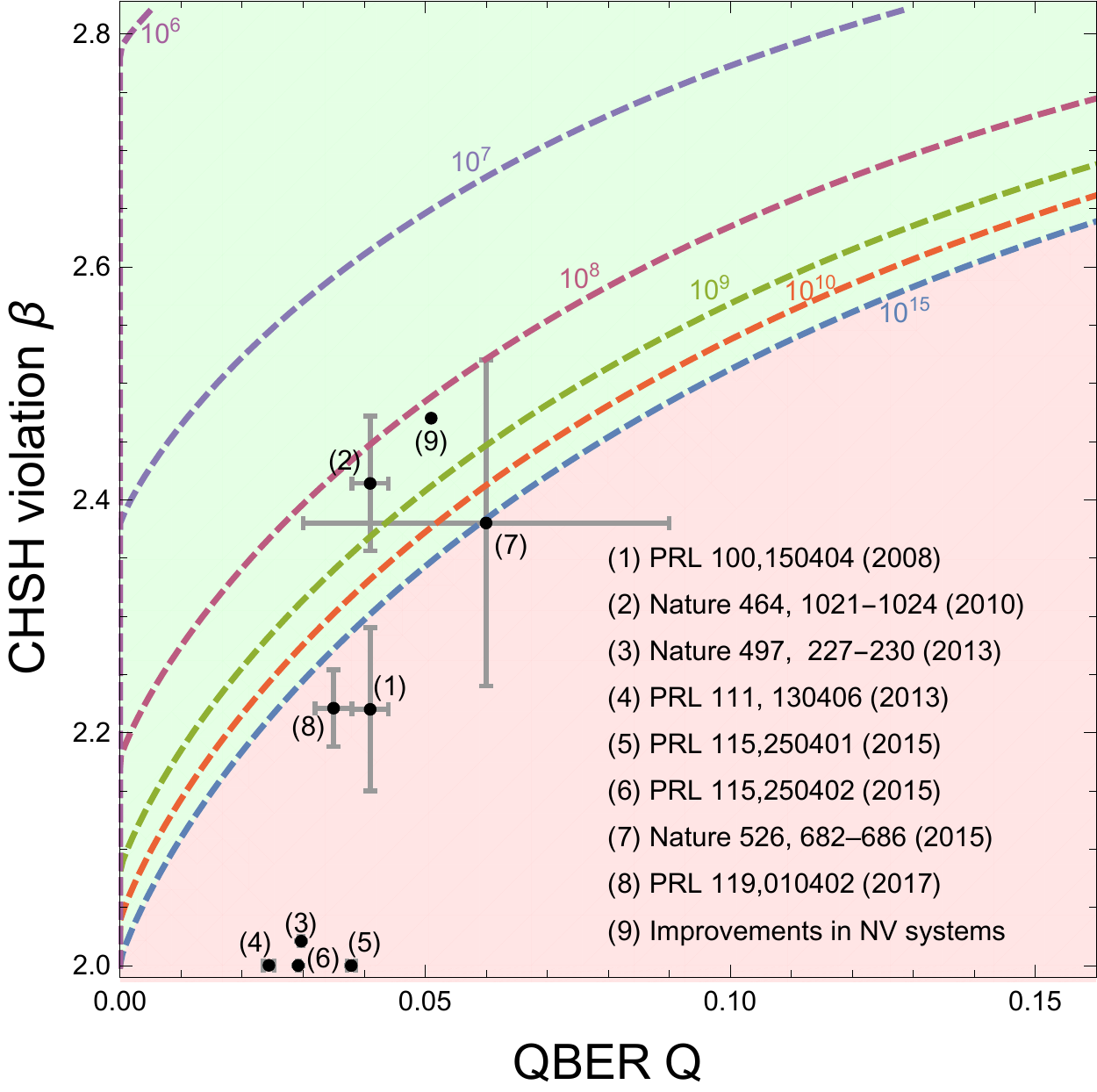}
      \end{center}
\end{figure}
\begin{figure}[H]
  \caption{Region of positive key rates for \emph{collective attacks}: The red area is the region of values of QBER ($Q$) and CHSH violation ($\beta$) for which a positive key rate cannot be reached with any number of rounds. In the green area, the dashed curves represents the minimum number of rounds required to get positive key rate. For parameters above each curve, a key rate can be extracted if the number of rounds is higher than specified in the curve. The points show the Bell violation and estimated QBER achieved by previous experiments (see Table \ref{tab:Bellparameter}). They, however, do not reflect the corresponding entanglement generation rates. Similarly to \cite{EATpublish}, we take $\epsilon^c_{DIQKD}=10 ^{-2}$ and $\epsilon^s_{DIQKD}=10^{-5}$.}\label{fig:Niid}
  \begin{center}
    \includegraphics[scale=0.65]{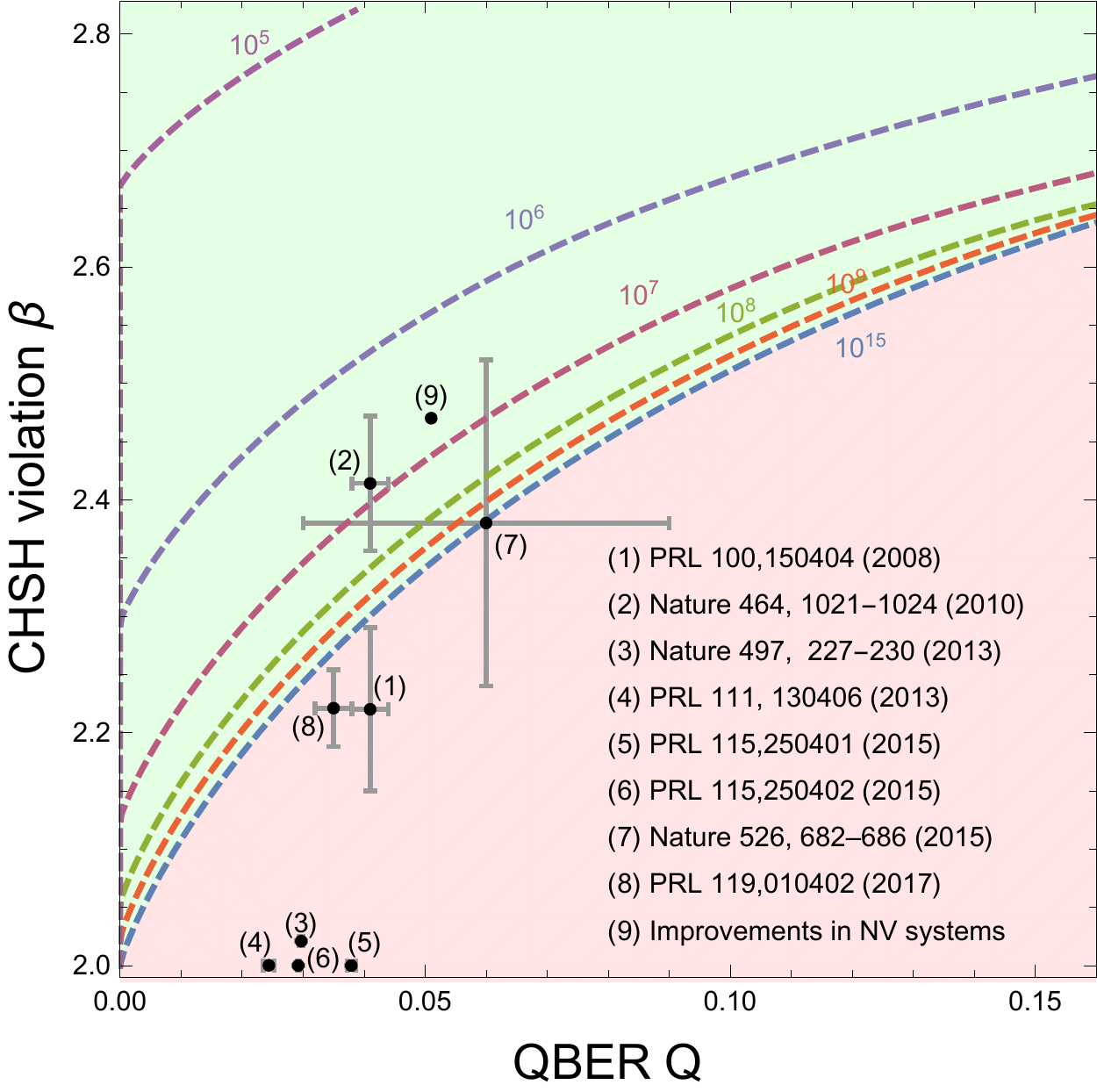}
    \end{center}
\end{figure}

\section{Discussion}\label{sec:discussion} 

Detection-loophole-free Bell tests between separated setups mark an important step towards the implementation of DIQKD. 
Progress towards extending Bell experiments to larger distances were also achieved, in particular by the Bell tests additionally closing the locality loophole.
However a DIQKD protocol has not yet been implemented.

In order to shed light on the experimental performance needed for DIQKD,
we have derived the 
key rates in the finite size regime as a function of the experimental parameters: CHSH violation $\beta$ and QBER $Q$.
For comparison of the key rates obtained in the finite regime for coherent and collective attacks, we have used 
as a benchmark an implementation where the maximally entangled state is subjected to depolarizing noise. Although the asymptotic key rates against collective attacks and general coherent attacks coincide, it is known that this is not the case in the finite regime. We find that, with the currently available tools, security against coherent attacks requires a minimum number of rounds about two orders of magnitude higher than what is necessary for security against collective attacks for realistic near-term parameters.

Here, we have focused on DIQKD protocols that use the CHSH inequality. So far the CHSH inequality is the one which leads to the best performance for a DIQKD protocol. 
The challenge in using other Bell inequalities is that, up to date, only non-tight lower bounds on the secure key rates can be derived.
Therefore, it is still an open question whether any other Bell inequality can outperform the CHSH, either in terms of maximum tolerable QBER, higher rates or lower minimum number of rounds required. 

Towards exploring the potential of different experimental platforms to implement DIQKD, we have analyzed the Bell violation and expected QBER of previously performed Bell tests and situated these parameters in the context of the derived key rates.
Figures \ref{fig:Ncoh} and \ref{fig:Niid} summarize this analysis.

For photonic systems, a DIQKD implementation is currently barred by the very low CHSH violation. To overcome this, a strong reduction of photon losses is required. 

Detection-loophole free Bell tests based on heralded entanglement schemes approach the allowed region, with the Bell test of Ref.~\cite{Pironio2010}, performed with trapped ions separated by 1 meter, even exhibiting parameters in the allowed region. These heralded schemes however suffer from low entangling rates resulting from photon losses. An increase in the entangling rates is expected to be achieved by improving collection efficiencies, e.g. by employing optical cavities. Moreover, with frequency downconversion these results can be extended to long ($\gg 1$~km) distances.
We illustrate that with near-term experimental improvements for NV centres, in combination with optical cavities for enhancing entangling rate, described in Section~\ref{sec:expNV}, a demonstration of DIQKD is achievable.

\section{Methods}\label{sec:methods} 

We now present the theoretical tools that allows us to derive the key rates for the device-independent quantum key distribution protocols, Protocol~\ref{prot:diqkd} and Protocol~\ref{prot:diqkdIID}.
We start by defining some quantities that are going to play an important role in the security proof and state in more details the security definition for device-independent quantum key distribution.

\subsection{Notation and definitions}

In cryptographic tasks, we are often interested in estimating what is the maximum probability with which an adversary can guess the value of a classical variable $A$\footnote{In QKD, for example, the classical variable is the string of bits that Alice holds after measuring her quantum systems.}. This is defined as the guessing probability $p_{\rm guess}$. In the general case where the adversary might have access to a quantum side information $E$, and therefore the state of interest is a $cq$-state (classical-quantum state) $\rho_{AE}$, the guessing probability is defined as:
\begin{eqnarray}\label{eq:pguess}
p_{\rm guess}(A|E)_{\rho}=\sup_{\DE{M_E^a}}\sum_{a}p(A=a)_{\rho} \mathrm{Tr} \de{M_E^a\rho_{E|A=a}},
\end{eqnarray}
where the supremum is taken over all POVMs $\DE{M_E^a}$ that can be performed on the system $E$.
The min-entropy of the classical variable $A$ conditioned on the quantum side information $E$ is then given by \cite{KRS09}
\begin{eqnarray}
H_{\min}(A|E)_{\rho}=- \log p_{guess}(A|E)_{\rho}.
\end{eqnarray}

A smoothed version of the min-entropy can also be defined.

\begin{definition}[Smooth min-entropy]\label{def:Hminsmooth}
	For a quantum state $\rho_{AE}$ and $\epsilon \in [0,1)$
	\begin{eqnarray}
	H_{\min}^{\epsilon}(A|E)_{\rho}=
	\sup_{\tilde{\rho}_{AE}\in \mathcal{B}^{\epsilon}(\rho_{AB})} H_{\min}(A|E)_{\tilde{\rho}},
	\end{eqnarray}
	where the supremum is taken over positive sub-normalized operators that are $\epsilon$-close to $\rho_{AB}$ in the purifying distance \cite{TomamichelBook}.
\end{definition}

The smoothing parameter $\epsilon$ allows us to restrict attention to  typical events (the ones that occur with probability higher than $1-\delta(\epsilon)$, where $\delta(\epsilon)$ is a function of the smoothing parameter).
 As a consequence, the smoothed min- and max-entropies (see \ref{Appendix:definitions} for definition)  have many nice properties and find an operational interpretation in many applications \cite{TomamichelBook,TomThesis}.

Other quantities of interest that will appear along the text are the conditional von-Neumann entropy, $H(A|E)_{\rho}$, and the conditional collision entropy $H_2(A|E)_{\rho}$. They are particular cases of the one-parameter family of entropies called sandwiched conditional R\'eyni entropies, first defined in Ref.~\cite{SandwichedReyni}.
\begin{definition}\label{def:alphaH}
	For any density operator $\rho_{AE}$ and for $\alpha \in [\frac{1}{2},1)\cup(1,\infty)$ the sandwiched $\alpha$-R\'eyni entropy of $A$ conditioned on $E$ is defined as
	\begin{eqnarray}\label{eq:Halpha}
	H_{\alpha}(A|E)_{\rho}:=\frac{1}{1-\alpha}\log\de{\mathrm{Tr}\De{\de{\rho_E^{\frac{1-\alpha}{2\alpha}}\rho_{AE} \rho_E^{\frac{1-\alpha}{2\alpha}} }^{\alpha}}},
	\end{eqnarray}
	where $\rho_E^{\frac{1-\alpha}{2\alpha}}$ is a short notation for $\id_A\otimes \rho_E^{\frac{1-\alpha}{2\alpha}}$.
	
	A variant can also be defined as
	\begin{eqnarray}\label{eq:Hup}
	H_{\alpha}^{\uparrow}(A|E)_{\rho}:=\sup_{\sigma_E\in \mathcal{S}}\frac{1}{1-\alpha}\log\de{\mathrm{Tr}\De{\de{\sigma_E^{\frac{1-\alpha}{2\alpha}}\rho_{AE} \sigma_E^{\frac{1-\alpha}{2\alpha}} }^{\alpha}}},
	\end{eqnarray}
	where $\mathcal{S}$ denotes the set of quantum states and the supremum is taken over density operators $\sigma_E$.
\end{definition}

The min- and max- entropy correspond to the extremal cases of definition (\ref{eq:Hup}) for $\alpha=\infty$ and $\alpha=\frac{1}{2}$ respectively.  For $\alpha\rightarrow 1$, definition (\ref{eq:Halpha}) and (\ref{eq:Hup}) coincide and one recover the standard conditional von-Neumann entropy.
Properties of the conditional  $\alpha$-R\'eyni entropies are presented in \ref{Appendix:definitions}.

\subsection{Security of DIQKD}\label{sec:security}

In order to determine what it means for a DIQKD protocol to be secure, we adopt the security definition used in \cite{DIEAT}.
This security definition follows the universally composable security definition for standard QKD protocols \cite{PR14}. However it is important to note that for the device-independent case composability was never proved and attacks proposed in Ref.~\cite{BCK13} show that composability is not achieved if the same devices are re-used for generation of a subsequent key. 

In the composably secure paradigm,  the security of a protocol is defined in terms of its distance to an ideal protocol  \cite{PR14,Canetti01}. Following this definition, given a protocol described by the completely positive and trace preserving (CPTP) map $\rm{diqkd}_{real}$, we say that the protocol is $ \epsilon_{DIQKD}^s$-secure for any $ \epsilon_{DIQKD}^s\geq \epsilon$ if: 
\begin{eqnarray}
\epsilon&:=\frac{1}{2}{\|\rm{diqkd}_{real}-\rm{diqkd}_{ideal}\|}_{\diamond}\\
&=\sup_{\rho_{ABE}}\frac{1}{2}{\|\rm{diqkd}_{real}(\rho_{ABE})-\rm{diqkd}_{ideal}(\rho_{ABE})\|}_1.\label{eq:sec}
\end{eqnarray}

Expression (\ref{eq:sec}) can be split into two terms that reflect independently the \emph{correctness} and the \emph{secrecy} of the protocol (see \cite{PR14}), given by Definitions \ref{def:correct} and \ref{def:secret}. Correctness is the statement that Alice and Bob share equal strings of bits at the end of the protocol. And secrecy states how much information the eavesdropper can have about their shared key.

Another requirement for a good DIQKD protocol is that there exist a realistic implementation that do not lead the protocol to abort almost all the time, \textit{i.e.}, the protocol should have some robustness. This is captured by the concept of \emph{completeness}.

\begin{definition}[Security]\label{def:security}
A DIQKD protocol is $(\epsilon^s_{DIQKD},\epsilon^c_{DIQKD},l)$-secure if
\begin{enumerate}
\item (Soundness) For any implementation of the protocol, either it aborts with probability greater than $1-\epsilon^s_{DIQKD}$ or an $\epsilon^s_{DIQKD}$-correct-and-secret key of length $l$ is obtained.
\item (Completeness) There exists an honest implementation of the protocol such that the probability of not aborting, $p(\Omega)$, is greater than $1-\epsilon^c_{DIQKD}$.
\end{enumerate}
\end{definition}

The correctness of the final key is ensured by the error correction step. During error correction, Alice sends to Bob a sufficient amount of information so that he can correct his raw key. If Alice and Bob do not abort in this step, then the probability that they end up with different raw keys is guaranteed to be very small. For the secrecy of the protocol, according to Definition~\ref{def:secret}, one needs to estimate how far the final state describing  Alice's key and the eavesdropper system is from a state where the eavesdropper is totally ignorant about Alice's key, see Eq.~(\ref{eq:secret}).
The formal security proof of quantum key distribution became possible due to the quantum Leftover Hashing Lemma \cite{RennerThesis,TSSR11} that quantifies the secrecy of a protocol as a function of a conditional entropy of the state before privacy amplification and the length of the final key.

\begin{theorem}[Leftover Hashing Lemma (\cite{RennerThesis}, Theorem 5.5.1)]\label{thm:leftoverH2}
Let $\rho_{A_1^nE}$ be a classical-quantum state and let  $\mathcal{H}$ be a 2-universal family of hash functions, from $\DE{0,1}^n$ to $\DE{0,1}^l$, that maps the classical $n$-bit string $A_1^n$ into $K_A$. Then
\begin{eqnarray}
{\|{\rho_{K_AHE}}-{\tau_{K_A}\otimes \rho_{HE}}\|}_1\leq 2^{-\frac{1}{2}\de{H_{2}^{\uparrow}(A_1^n|E)_{{\rho}}-l}}.
\end{eqnarray}
\end{theorem}
For the proof of the Leftover Hashing Lemma we refer to Ref.~\cite{RennerThesis}. In Ref.~\cite{RennerThesis}, it was shown 
that the Leftover Hashing lemma can also be formulated in terms of the smooth min-entropy, and the price to pay is only a linear term in the security parameter\footnote{In Ref.~\cite{RennerThesis},  the leftover hash lemma was formulated with the smooth min-entropy defined as a supremum over states that are $\epsilon$-close to $\rho$ in the trace norm. The proof of Theorem~\ref{thm:leftoverHmin}, with the smooth min-entropy defined according to Definition~\ref{def:alphaH}, can be found in Ref.~\cite{RigorousQKD}.}.

\begin{theorem}[Leftover Hashing Lemma with smooth min-entropy  \cite{RennerThesis,RigorousQKD}]\label{thm:leftoverHmin}
Let $\rho_{A_1^nE}$ be a classical-quantum state and let $\mathcal{H}$ be a 2-universal family of hash functions, from $\DE{0,1}^n$ to $\DE{0,1}^l$, that maps the classical $n$-bit string $A_1^n$ into $K_A$. Then
\begin{eqnarray}
{\|{\rho_{K_AHE}}-{\tau_{K_A}\otimes \rho_{HE}}\|}_1\leq 2^{-\frac{1}{2}\de{H_{\min}^{\epsilon}(A_1^n|E)_{{\rho}}-l}}+2 \epsilon.
\end{eqnarray}
\end{theorem}

Given the Leftover Hash Lemma, stated in Theorems \ref{thm:leftoverH2} and \ref{thm:leftoverHmin}, and the definition of secrecy, Definition \ref{def:secret}, we can now express the length of a secure key  as a function of the entropy of Alice's raw key conditioned on Eve's information before privacy amplification.

\begin{theorem}[Key length]\label{thm:keysize}
Let $p(\Omega)$ be the probability that the DIQKD protocol does not abort for a particular implementation. 
If the length of the key generated after privacy amplification is given by
\begin{eqnarray}\label{eq:keysizeH2}
l&=H_{2}^{\uparrow}(A_1^n|E)_{\rho_{|\Omega}}-2\log\de{\frac{1}{2\epsilon_{PA}}}.
\end{eqnarray}
then  the DIQKD protocol is  $\epsilon_{PA}$-secret.

We can also express the key length in terms of the smooth min-entropy, where if $l$ satisfies
\begin{eqnarray}\label{eq:keysize}
l&=H_{\min}^{\epsilon_s/p(\Omega)}(A_1^n|E)_{\rho_{|\Omega}}-2\log\de{\frac{p(\Omega)}{2\epsilon_{PA}}}\\
&\geq H_{\min}^{\epsilon_s/p(\Omega)}(A_1^n|E)_{\rho_{|\Omega}}-2\log\de{\frac{1}{2\epsilon_{PA}}},
\end{eqnarray}
then  the DIQKD protocol is  $(\epsilon_{PA}+\epsilon_s)$-secret.
\end{theorem}

We see that the leftover hashing lemma expressed in terms of smooth min-entropy only leads to an extra $\epsilon_s$ term in the security parameter. However, the smooth min-entropy can be much larger than the $2$-R\'enyi entropy $H_2^{\uparrow}$ and, therefore, it is advantageous to lower bound the key by the smooth min-entropy.

\subsection{Security analysis}\label{sec:techniques}

In the previous section we have seen that in order to determine the length of a secret key generated by a particular protocol one needs to estimate the (smooth-min or 2-R\'enyi) entropy  of Alice's string conditioned on all the information available to the eavesdropper before privacy amplification. 
Now, in order to estimate this quantity for a DIQKD protocol one faces two main challenges:

\begin{itemize}
\item How to evaluate the entropy of a very long string of bits? 
\item How to evaluate the one-round entropy in the device-independent scenario? 
\end{itemize}

In Section \ref{sec:1roundreduction} we present the theoretical tools that allow to reduce the problem of evaluating the entropy of a string of bits to the evaluation of a single round. Moreover,
in the DI scenario we do not want to make any assumptions over the underlying quantum state and measurement devices. In Section \ref{sec:Hevaluation} we present a tight bound derived in \cite{PAB07,PAB09} for the one round conditional von Neumann entropy of protocols where Alice and Bob test the CHSH inequality. Moreover we explore further this bound to prove a tight bound on the single round conditional collision entropy as a function of the CHSH violation.

\subsubsection{Reducing the problem to the estimation of one round.}\label{sec:1roundreduction}

We now present the techniques that allow to reduce the evaluation of the entropy $H_{\min}^{\epsilon_s/p(\Omega)}(A_1^n|E)_{\rho_{|\Omega}}$ to the estimation of the conditional von Neumann entropy of a single round for the two adversarial scenarios under consideration, collective attacks and coherent attacks.
Moreover, for the IID scenario, \textit{i.e.} when the eavesdropper is assumed to be restricted to collective attacks, we show how to break the analysis of the entropy $H_{2}^{\uparrow}(A_1^n|E)_{\rho_{|\Omega}}$ into single rounds evaluation. 

\subsubsection*{The IID scenario (collective attacks).}

When we restrict the eavesdropper to collective attacks, we are assuming that, even though she can perform an arbitrary operation in her quantum side information, the state distributed by the source and the behavior of Alice's  and Bob's devices are the same in every round of the protocol. This implies that after $n$ rounds, the state shared by Alice, Bob and Eve is $\rho_{A_1^nB_1^nE}=\rho_{ABE}^{\otimes n}$. 
In this case, the quantum asymptotic equipartition property (AEP) \cite{TCR09} allows to break the conditional smooth min-entropy of state $\rho_{AE}^{\otimes n}$ into $n$ times the conditional von Neumann entropy of the state $\rho_{AE}$. 

\begin{theorem}[Asymptotic equipartition property \cite{TCR09}]\label{thm:AEP}
Let $\rho=\rho_{AE}^{\otimes n}$ be an i.i.d. state. Then for $n\geq \frac{8}{5}\log\frac{2}{\epsilon^2}$
\begin{eqnarray}
H_{\min}^{\epsilon}(A_1^n|E_1^n)_{\rho_{AE}^{\otimes n}}\geq nH(A|E)_{\rho_{AE}}-\sqrt{n}\,\delta(\epsilon,\eta)
\end{eqnarray}
and similarly
\begin{eqnarray}
H_{\max}^{\epsilon}(A_1^n|E_1^n)_{\rho_{AE}^{\otimes n}}\leq nH(A|E)_{\rho_{AE}}+\sqrt{n}\,\delta(\epsilon,\eta)
\end{eqnarray}
where $\delta(\epsilon, \eta)=4\log \eta \sqrt{\log\frac{2}{\epsilon^2}}$ and $\eta = \sqrt{2^{-H_{\min}(A|E)_{\rho_{AE}}}}+\sqrt{2^{H_{\max}(A|E)_{\rho_{AE}}}}  +1$.
\end{theorem}

The quantum AEP is a generalization to quantum systems of the classical statement that, in the limit of many repetitions of a random experiment, the output sequence is one from the typical set.
Therefore, under the assumption of collective attacks, the quantum AEP reduces the problem of estimating  the key rate of a string of $n$ bits to the problem of bounding the one-round conditional von Neumann entropy. 
We remark that  the AEP implies an additional term, proportional to $\sqrt{n}$, which is significant for the finite regime analyses.\vspace{1em}

In Section \ref{sec:security}, we have seen that the left-over hashing lemma can also be stated in the terms of the 2-R\'eyni conditional entropy  $H^{\uparrow}_2(A|E)_{\rho}$. A useful property of the conditional $H^{\uparrow}_\alpha$  entropies is additivity \cite{TomamichelBook} (see \ref{Appendix:definitions} Property \ref{PropH}(\ref{propAdd})), which implies the following lemma.

\begin{lemma}\label{lem:H2}
Let $\rho=\rho_{AE}^{\otimes n}$ be an i.i.d. state. Then
\begin{equation}\label{eq:breakH2}
H^{\uparrow}_2(A_1^n|E_1^n)_{\rho_{AE}^{\otimes n}}=n   H^{\uparrow}_2(A|E)_{\rho_{AE}}\geq n   H_2(A|E)_{\rho_{AE}},
\end{equation}
where $H_2(A|E)_{\rho_{AE}}$ is denoted collision entropy.
\end{lemma}

Validity of Lemma~\ref{lem:H2} can be seen from the following: equality in (\ref{eq:breakH2}) follows from the additivity property of $H^{\uparrow}_\alpha$  entropies, Property \ref{PropH}(\ref{propAdd}) in \ref{Appendix:definitions}, and the inequality follows from the definition of $\alpha$-R\'enyi entropies, Definition~\ref{def:alphaH}.\vspace{1em}

Therefore, for collective attacks one can break the analysis 
into the evaluation of a single-round entropy by using both, the formulation of the left-over hashing lemma in terms of the smooth-min entropy, Theorem~\ref{thm:leftoverHmin}, and in terms of the $2$-R\'enyi entropy, Theorem~\ref{thm:leftoverH2}. The possible advantage of using Lemma~\ref{lem:H2} over the AEP, Theorem~\ref{thm:AEP}, is that no extra overhead term  $\mathcal{O}(\sqrt{n})$ is gained due to the additive property of the 2-R\'eyni conditional entropy  $H^{\uparrow}_2(A|E)_{\rho}$. 
However, in general the von Neumann entropy can be much larger than the collision entropy, and this trade-off has to be taken into account. 
We remark that, for protocols based on other Bell inequalities, the techniques used for deriving Theorem~\ref{thm:rateH2} can be advantageous for collective attack analysis. This is due to the fact that for other Bell inequalities there is no known technique to directly bound the conditional von-neumann entropy and a good bound on the min-entropy can be found using semidefinite-programming techniques (see Section~\ref{sec:Hevaluation}).

\subsubsection*{The fully DI scenario (coherent attacks).}

In the fully device-independent scenario the eavesdropper can perform a general coherent attack, and the state shared by the parties may not be of the form $\rho_{ABE}^{\otimes n}$. Therefore, the tools presented in the previous section are not applicable in this scenario. 
In standard QKD, de Finetti techniques \cite{RennerThesis,KR05, CKR09} allow one to extend the proofs against collective attacks to coherent attacks for protocols that present some symmetry. The price to pay is an overhead term $\mathcal{O}(\sqrt{n})$  whose pre-factor depends on the dimension of the underlying system. However, in the device-independent scenario, we do not want to make assumptions on the dimension of the underlying system. Moreover, symmetry of the protocol is not guaranteed, as we do not know the behaviour of the measurement devices. Therefore, de Finetti techniques cannot be used to straightforwardly extend the security proofs against collective attacks to coherent attacks in the device-independent scenario.

Recently, this problem was overcome by the entropy accumulation theorem (EAT) \cite{EATpublish,EAT}. In this section, we state the entropy accumulation theorem, which allows to break the entropy $H_{\min}^{\epsilon_s/p(\Omega)}(A_1^n|E)_{\rho_{|\Omega}}$ into the entropy of single rounds and therefore extends proofs against collective attacks to coherent attacks. 

An important ingredient in the formulation of the EAT is the concept of \emph{min-/max-tradeoff function} of a channel.

\begin{definition}\label{def.fmin}
Let  $\mathcal{N}_i$ be a CPTP map that maps $R_{i-1}$ to $\hat{A}_i\hat{B}_iC_iR_i$, where $\hat{A}_i$, $\hat{B}_i$ and $C_i$ are classical registers and the value of $C_i$  can be inferred from $\hat{A}_i$ and $\hat{B}_i$. Let $\vec{q}$ denote a probability distribution on the possible values the random variable $C_i$ can assume.
The min- and max-tradeoff functions for the channel $\mathcal{N}_i$ are defined as:
\begin{eqnarray}
f_{\min}(\vec{q})&\leq \inf_{\sigma \in \Sigma_i(\vec{q})}H(\hat{A}_i|\hat{B}_iR)_{\sigma},\\
f_{\max}(\vec{q})&\geq \sup_{\sigma \in \Sigma_i(\vec{q})}H(\hat{A}_i|\hat{B}_iR)_{\sigma},
\end{eqnarray}
where 
\begin{eqnarray}
\Sigma_i(\vec{q})=\DE{\sigma_{C_i\hat{A}_i\hat{B}_iR_iR}=(\mathcal{N}_i\otimes I_R)(\omega_{R_{i-1}R})|\sigma_{C_i}=\vec{q}},
\end{eqnarray}
and the infimum and supremum are set to $+\infty, -\infty$, respectively, if the set $\Sigma_i(\vec{q})$ is empty.
\end{definition}

 Definition~\ref{def.fmin} states that the min-(max-)tradeoff function is a lower (upper) bound on the conditional von Neumann entropy $H(\hat{A}_i|\hat{B}_iR)_{\sigma}$ of a final state $\sigma_{C_i\hat{A}_i\hat{B}_iR_iR}$, for all states
that result from the
action of the channel $\mathcal{N}_i$ on an arbitrary initial state and exhibit a particular distribution $\vec{q}$ over the classical variable $C_i$, where $R$ is a side information. In particular, for a DIQKD protocol, where we are testing the CHSH inequality, the variable $\hat{A}_i$ can be the outputs of Alice and Bob in round $i$, $\hat{A}=\DE{A_i,B_i}$. The variable $\hat{B}_i$ can be the inputs of Alice and Bob together with the variable that determines whether the round is a test round or a key generation round, $\hat{B}_i=\DE{X_i,Y_i,T_i}$. And $R$ can represent any quantum side information $E$ that the eavesdropper holds. We will then be interested in defining a variable $C_i$ that assumes value 1 if the condition of  the CHSH game is satisfied (\textit{i.e.} if the outputs of Alice and Bob satisfy $A_i+B_i=X_i\cdot Y_i$), 0 if it is not satisfied and we attribute the value $\bot$ if the inequality was not tested in that round (\textit{i.e.} if $T_i=0$, the key generations rounds). Now the distributions $\vec{q}=(q(0),q(1),q(\bot))$ of interest are the ones that achieve a winning probability $\omega$ for the CHSH game, \textit{i.e.} $\frac{q(1)}{1-q(\bot)}=\omega$. The EAT channel $\mathcal{N}_i$ represents local maps that, according to the value of $T_i$, generate the variables $X_i$, $Y_i$ randomly and independently, and then generate the outcomes $A_i$ and $B_i$. Finally, the set of states $\Sigma_i(\vec{q})$ of interest are all the states resulting from the action of this channel in an arbitrary state and exhibiting a violation $\beta=8\omega -4$ for the CHSH inequality. For a more detailed description of the EAT channel associated to Protocol~\ref{prot:diqkd}, we refer the reader to \cite{DIEAT,EATpublish}.

We now state the entropy accumulation theorem.

\begin{theorem}[The entropy accumulation theorem (EAT) \cite{EAT}]\label{thm.EAT}
For an event $\Omega$ that happens with probability $p(\Omega)$, and for $t$ such that $f_{\min}(freq (c_1^n))\geq t$ $\forall$ $c_1^n \in \Omega$, it holds that
\begin{eqnarray}
H_{\min}^{\epsilon}(A_1^n|B_1^nE)_{\rho_{|\Omega}}>nt-\nu\sqrt{n}
\end{eqnarray}
and similarly, for $t'$ such that  $f_{\max}(freq (c_1^n))\leq t'$  $\forall$ $c_1^n \in \Omega$,
\begin{eqnarray}
H_{\max}^{\epsilon}(A_1^n|B_1^nE)_{\rho_{|\Omega}}<nt'+\nu\sqrt{n}
\end{eqnarray}
with
\begin{eqnarray}\label{eq.EAToverhead}
\nu=2\de{\log\de{1+2d_A}+\lceil\|\nabla f\|_{\infty}\rceil}\sqrt{1-2\log\de{\epsilon_s\cdot p(\Omega)}}
\end{eqnarray}
 for $f$ equals to $f_{\min}$ and $f_{\max}$ respectively .
\end{theorem}

Analogous to the AEP, the entropy accumulation theorem allows us to break the entropy of the string of bits into the entropy of a single round. Note, however, that this single-round entropy does not refer to the real entropy of each round of the protocol, but is evaluated over the hypothetical states that would achieve the observed violation. 
It is important to remark that a crucial assumption in the EAT \cite{EAT,EATpublish} is that some of the variables of interested satisfy what is called the Markov condition. This is the case for QKD protocols performed sequentially. For definition and discussion of the implications of the Markov condition, see \cite{EAT}.

\subsubsection{Estimating the one-round entropy.}\label{sec:Hevaluation}

Now that we have reduced the evaluation of the secret key length to the estimation of the conditional von Neumann entropy of a single round, we are ready to face the next  challenge: How to estimate the single round entropy without any assumptions on the quantum states and behavior of the measurement devices.

\subsubsection*{The CHSH scenario:}

The CHSH scenario \cite{CHSH}, where Alice and Bob each perform one among two possible binary measurements, is significantly simpler than other Bell scenarios. Due to the fact that the CHSH inequality has only two binary inputs per party, a strong result \cite{Masanes06, Tsirelson93} states that the description of any realization of a CHSH experiment can be decomposed into subspaces of dimension two, where projective measurements are performed in each subspace.
This allows one to restrict the analysis to qubits, which significantly simplifies the problem. Exploring these nice properties, a tight bound on the von Neumann entropy of Alice's outcome conditioned on Eve's information, as a function of the CHSH violation, was derived in \cite{PAB07,PAB09}.

\begin{lemma}\label{lem:Hchsh}
Given that Alice and Bob share a state $\rho_{AB}$ that achieves a violation $\beta$ for the CHSH inequality, it holds that
\begin{equation}\label{eq:Hbeta}
H(A|E)_{\rho}\geq 1-h\de{\frac{1}{2}+\frac{1}{2}\sqrt{\de{\frac{\beta}{2}}^2-1}}.
\end{equation}
\end{lemma}
\vspace{1em}

In Section~\ref{sec:1roundreduction} we have seen that for collective attacks the key rate can also be estimated by the single round collision entropy. And due to the additivity property of $H_2^{\uparrow}$, no overhead $\sqrt{n}$ term is present. Therefore, this analysis can potentially lead to an advantage with respect to the minimum number of rounds required for positive key rate. 
The conditional collision entropy satisfies the following relation \cite[Corollary 5.3]{TomamichelBook} 
\begin{equation}\label{eq:H2Hmin}
H_2(A|E)_{\rho}\geq H_{\min}(A|E)_{\rho}.
\end{equation}
And a lower bound for the conditional min-entropy as a function of the Bell violation was derived in \cite{HminBell}:
\begin{equation}\label{eq:Hminbeta}
H_{\min}(A|E)_{\rho}\geq -\log\de{\frac{1}{2} +\frac{1}{2}\sqrt{2-\frac{\beta^2}{4}}}.
\end{equation}

Therefore expression (\ref{eq:Hminbeta}) can be used to bound the conditional collision entropy as a function of the violation $\beta$. We now prove that this bound is actually tight. 

\begin{theorem}\label{thm:H2}
There exist a state $\rho^*_{AB}$ and measurements for Alice and Bob such that, $\rho^*_{AB}$ achieves violation $\beta$ and the collision entropy of Alice's output $A$ conditioned on Eve's quantum information $E$ is 
\begin{equation}
H_2(A|E)_{\rho^*}= -\log\de{\frac{1}{2} +\frac{1}{2}\sqrt{2-\frac{\beta^2}{4}}}.
\end{equation}
\end{theorem}

The proof of Theorem \ref{thm:H2} is presented in \ref{Appendix:H2}.
Theorem \ref{thm:H2} together with relations (\ref{eq:H2Hmin}) and (\ref{eq:Hminbeta}) imply a tight lower bound for the conditional collision entropy as a function of the CHSH violation $\beta$.
In Figure \ref{HxH2} we plot $H(A|E)$ and $H_{2}(A|E)$ as a function of the violation $\beta$. One can see that the points of maximum and minimum entropy (corresponding to maximal violation $\beta=2\sqrt{2}$ and no violation, respectively) coincide, but for intermediate values of $\beta$ the conditional collision entropy is smaller than the conditional von Neumann entropy.

\begin{figure}[H]\label{HxH2}
\begin{minipage}{20em}
\includegraphics[scale=0.7]{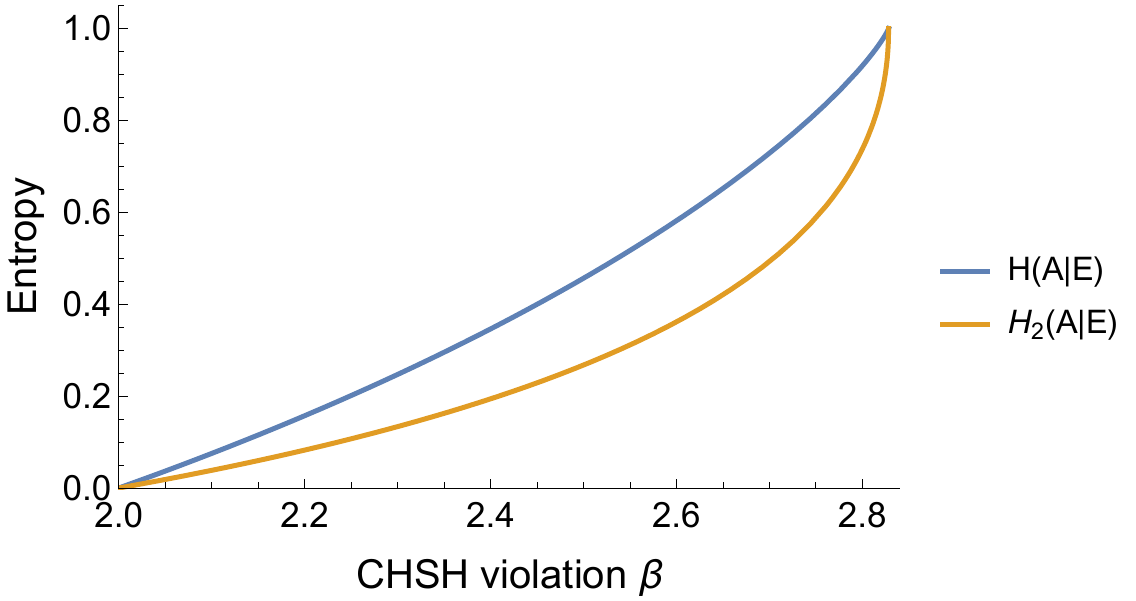}
\end{minipage}\hfill
\begin{minipage}{19em}
\caption{Graph illustrating the difference of the conditional von Neumann entropy $H(A|E)$ and the conditional collision entropy $H_{2}(A|E)$ as a function of the  CHSH violation $\beta$.}\label{Fig:HxHmin}
\end{minipage}
\end{figure}

\subsubsection*{Other Bell inequalities and the min-entropy estimation:}

The use of different Bell inequalities has proved to be advantageous in different taks. 
For example, a tilted CHSH inequality was used to certify maximal randomness in states arbitrarily close to separable \cite{AMP12}, and   inequalities with more inputs and outputs have shown to exhibit higher noise robustness \cite{CGLMP}.
Therefore it is natural to ask whether other Bell inequalities can also bring advantage to the task of device-independent quantum key distribution.

By considering an arbitrary Bell inequality, one faces the problem that the techniques used to bound the conditional von Neumann entropy as a function of the CHSH violation do not apply.
Indeed, the proof 
of Lemma~\ref{lem:Hchsh}  is highly based on the fact that one can reduce the analysis to qubits. 
In fact, very few results are known on tight bounds for the conditional von Neumann entropy as a function of the Bell violation for other inequalities. In \cite{Rotem17} a bound was derived for a family of inequalities denoted measurement-device-dependent inequalities \cite{PRB14}, which are very suitable for the task of randomness amplification.  In \cite{JeremyCKA} a tight bound was derived as a function of the violation of the multipartite MABK inequality \cite{Mer90,Ard92,BK93}. However in these two cases the proof is based on a reduction to the CHSH inequality.

In general, the conditional von Neumann entropy can be lower bounded by the conditional min-entropy
\begin{equation}
H(A|E)_{\rho}\geq H_{\min}(A|E)_{\rho}.
\end{equation}
The advantage of looking at the conditional min-entropy is that it can be computed as a function of the Bell violation by a semi-definite programming \cite{HminBell}.
The idea is that in order to estimate the min-entropy one can upper bound the guessing probability, $p_{\rm guess}$ (see Eq.~(\ref{eq:pguess})), of the eavesdropper.
This problem  can then be expressed as an optimization over  probability distributions, which is exactly the information available in the device-independent scenario.
As shown in Ref. \cite{HminBell}, for any Bell inequality, an upper bound on the $p_{\rm guess}$ can be obtained by a semidefinite programming making use of the NPA-hierarchy \cite{NPA, NPA2}.

Lower bounding the conditional von-Neumann entropy by the min-entropy might be far from optimal.
For example, for the CHSH inequality we have that the conditional von Neumann entropy as a function of the violation is much larger than the conditional min-entropy, as illustrated in Fig.~\ref{Fig:HxHmin} (recall that, in Theorem~\ref{thm:H2}, $H_{\min}(A|E)_{\rho}$ was shown to be a tight bound on $H_2(A|E)_{\rho}$ as a function of the CHSH violation).
By making use of the tight bound on the conditional von Neumann entropy, eq.~(\ref{eq:Hbeta}), 
 one can prove security for DIQKD up to $7.1\%$ of QBER \cite{PAB09}, whereas using the min-entropy, eq.~(\ref{eq:Hminbeta}), security can only be guaranteed up to a QBER of $5.2 \%$ \cite{HminBell}.

 It is still an open problem whether any other Bell inequality can lead to better performance for DIQKD than the CHSH inequality. Recently, an extensive analysis of the performance of different Bell inequalities for the task of randomness expansion was presented in \cite{Brown18}.

\subsubsection{Key rates.}
 
The techniques presented in Sections \ref{sec:1roundreduction} and \ref{sec:Hevaluation} allows us to establish the length of a secure key that can be extracted as a function of the CHSH violation $\beta$ and QBER $Q$.

For coherent attacks, the entropy accumulation theorem (Theorem~\ref{thm.EAT}) and the tight lower bound on the conditional von Neumann entropy (Lemma~\ref{lem:Hchsh}) are the key tools to establish Theorem~\ref{thm:rateEAT}. The complete proof of Theorem~\ref{thm:rateEAT} includes several intermediate steps, and is presented in details in \ref{Appendix:rEAT}. 

For collective attacks, the key ingredients to derive Theorem~\ref{thm:rateIID} are the asymptotic equipartition property (Theorem~\ref{thm:AEP}) and Lemma~\ref{lem:Hchsh}. 
A detailed proof of Theorem~\ref{thm:rateIID} is presented in \ref{Appendix:rIID}. 
We have also presented a different technique of breaking the entropy of Alice's string into the entropy of single rounds in the IID scenario, namely by making use use of the additivity of 2-R\'eyni entropy, Lemma~\ref{lem:H2}. This technique, together with Theorem~\ref{thm:H2} leads to Theorem~\ref{thm:rateH2}. A detailed proof of Theorem~\ref{thm:rateH2} can be found in \ref{Appendix:rH2}.

\section*{Acknowledgments} 
We thank Victoria Lipinska, Mark Steudtner and Maximilian Ruf for helpful feedback on the manuscript, and Marissa Giustina for sharing with us her PhD thesis. 
This work was supported by the Dutch Technology Foundation
(TTW), the European Research Council through a Synergy Grant (RH) and a Starting Grant (SW), the Royal Netherlands Academy of Arts and Sciences and Ammodo through an Ammodo KNAW Award (RH), and the Netherlands Organisation for Scientific Research (NWO) through a VICI grant (RH), a VIDI grant (SW) and a Zwaartekracht QSC grant (RH and SW).

\section*{References}
\bibliographystyle{ieeetr}
\bibliography{biblio}

\appendix

\section{Definitions}\label{Appendix:definitions}

In this Appendix we present some properties of the conditional sandwiched $\alpha$-R\'eyni entropies \cite{SandwichedReyni}, Definition \ref{def:alphaH}, and the smoothed entropies  that are used for the security proof.

\begin{prop}\label{PropH} The conditional $\alpha$-R\'enyi  entropies satisfy:
	\begin{enumerate}
		\item\label{propDataProcess} \textbf{Data processing} (\cite{TomamichelBook} Corollary 5.1):  Let 
		$\tau_{AB'}=I_A\otimes \mathcal{E}_B(\rho_{AB})$, where $\mathcal{E}_B$ is a CPTP($B,B'$) channel, then
		\begin{eqnarray}
		H_{\alpha}(A|B)_{\rho}\leq H_{\alpha}(A|B')_{\tau} \;\;\mathrm{ and }\;\; H_{\alpha}^{\uparrow}(A|B)_{\rho}\leq H_{\alpha}^{\uparrow}(A|B')_{\tau}. 
		\end{eqnarray}
		\item \label{propAdd} \textbf{Additivity} (\cite{TomamichelBook} Corollary 5.2): For $\rho_{AB}\otimes \tau_{A'B'}$ it holds that
		\begin{eqnarray}
		H_{\alpha}^{\uparrow}(AA'|BB')_{\rho\otimes\tau}=H_{\alpha}^{\uparrow}(A|B)_{\rho}+ H_{\alpha}^{\uparrow}(A'|B')_{\tau}.
		\end{eqnarray}
		\item \textbf{Entropy of classical information}(\cite{TomamichelBook} Lemma 5.3): For $\rho_{ABX}$ classical in $X$ 
		\begin{eqnarray}
		H_{\alpha}(XA|B)_{\rho}\geq H_{\alpha}(A|B)_{\rho} \;\;\mathrm{ and }\;\; H_{\alpha}^{\uparrow}(XA|B)_{\rho}\geq H_{\alpha}^{\uparrow}(A|B)_{\rho}. 
		\end{eqnarray}
		\item\label{prop:Dimsystem} \textbf{Conditioning on classical information} (see \cite{TomamichelBook} Lemma 5.4): For $\rho_{ABX}$ classical in $X$, 
		\begin{eqnarray}
		H_{\alpha}^{\uparrow}(A|XB)&\geq H_{\alpha}^{\uparrow}(A|B)-\log \de{\rm{rank} (\rho_X)}\\
		&\geq H_{\alpha}^{\uparrow}(A|B)-\log |X|,
		\end{eqnarray}
		where $\rm{rank} (\rho_X)$ is the rank of matrix $\rho_X$ and $|X|$ is the dimension of system $X$.
		\item\label{prop:condClassic} \textbf{Conditioning on classical information} (see \cite{TomamichelBook} Proposition 5.1): Let $\rho_{ABX}=\sum_xp_x\rho_{AB}^x\otimes \ketbra{x}{x}$ then, 
		\begin{eqnarray}
		H_{\alpha}(A|BX)_{\rho}&=\frac{1}{1-\alpha} \log\de{\sum_xp(X=x)2^{\de{ (1-\alpha)H_{\alpha}(A|BX=x)_{\rho}}}},\\
		H_{\alpha}^{\uparrow}(A|BX)_{\rho}&=\frac{\alpha}{1-\alpha} \log\de{\sum_xp(X=x)2^{\de{ \frac{1-\alpha}{\alpha}H_{\alpha}^{\uparrow}(A|BX=x)_{\rho}}}}.
		\end{eqnarray}
			And for the conditional von Neumann it holds that
					\begin{eqnarray}
		H(A|BX)_{\rho}&=\sum_xp(X=x)H(A|BX=x)_{\rho}.
		\end{eqnarray}
		\gm{\item\label{prop:condState} \textbf{Entropy of the conditioned state} (see \cite{EAT} Lemma B.5): Let $\rho_{ABX}=\sum_xp_x\rho_{AB|x}$ then, 
		\begin{eqnarray}
		H_{\alpha}^{\uparrow}(A|B)_{\rho_{AB|x}}&\geq H_{\alpha}^{\uparrow}(A|B)_{\rho}-\frac{\alpha}{\alpha -1}\log\de{\frac{1}{p_x}}.
		\end{eqnarray}}
	\end{enumerate}
\end{prop}

In Property~\ref{PropH}.(\ref{prop:Dimsystem}), the relation $H_{\alpha}^{\uparrow}(A|XB)\geq H_{\alpha}^{\uparrow}(A|B)-\log |X|$ was stated in \cite{TomamichelBook}. We remark that the middle inequality follows from the fact that $H_{\alpha}^{\uparrow}(A|XB)$ is invariant under local isometries. Therefore if $X'=\mathcal{V}(X)$ is a full rank operator where $\mathcal{V}(\cdot)$ is an isometry, we have that
\begin{eqnarray}
H_{\alpha}^{\uparrow}(A|XB)=H_{\alpha}^{\uparrow}(A|X'B)\geq H_{\alpha}^{\uparrow}(A|B)-\log |X'|
\end{eqnarray}
and since $\mathcal{V}(\cdot)$ is an isometry $|X'|=\rm{rank}(\rho_X)$.

The min- and max- entropy are the particular extreme  cases of $H_{\alpha}^{\uparrow}$ for $\alpha=\infty$ and $\alpha=\frac{1}{2}$ respectively.  For $\alpha\rightarrow 1$ one recovers the standard conditional von-Neumann entropy.
The smoothed min- and max-entropies are defined as an optimization over operators that are $\epsilon$-close, in the purified distance, to the state of interest. This optimization takes into account also operators that are sub-normalized, \textit{i.e.} positive operators with trace smaller than 1.

\begin{definition}[Smoothed entropies \cite{TomamichelBook}]\label{def:Hsmooth}
	Let $\rho_{AB}$ be a quantum state and $\epsilon \geq 0$. The smooth min-entropy of system A conditioned on B is defined as
	\begin{eqnarray}
	H_{\min}^{\epsilon}(A|B)_{\rho}=\max_{\tilde{\rho}_{AB}\in \mathcal{B}^\epsilon(\rho_{AB})}H_{\min}(A|B)_{\tilde{\rho}}.
	\end{eqnarray}
	The smooth max-entropy is 	
	\begin{eqnarray}
	H_{\max}^{\epsilon}(A|B)_{\rho}=\min_{\tilde{\rho}_{AB}\in \mathcal{B}^\epsilon(\rho_{AB})}H_{\max}(A|B)_{\tilde{\rho}}.
	\end{eqnarray}
\end{definition}

In Definition~\ref{def:Hsmooth}, $\mathcal{B}^\epsilon(\rho_{AB})$ is an $\epsilon$-ball of sub-normalized operators around state $\rho_{AB}$ defined in terms of the purified distance.

\begin{definition}[Purified distance \cite{TomamichelBook}]
	For sub-normalized positive operators $X$ and $Y$, \textit{i.e.} $X,Y \geq 0$ and $\Tr(X)\leq 1, \Tr(Y)\leq 1$, the purified distance is given by
	\begin{equation}
	\mathcal{D}(X,Y)=\sqrt{1-F_*(X,Y)},
	\end{equation}
	where $F_*(\cdot,\cdot)$ is the generalized fidelity, defined  as
	\begin{equation}
	F_*(X,Y)=\de{\Tr|\sqrt{X}\sqrt{Y}|+\sqrt{(1-\Tr \rho)(1-\Tr(Y))}}^2.
	\end{equation}
\end{definition}

The smoothed entropies satisfy several chain rules. Some of them are stated below. A more complete list of chain rule relations can be found in \cite{TomamichelBook, Vit13}.

\begin{prop}[Chain rules for the smooth min-entropy]\label{prop:chain} 
The smooth min-entropy satisfy the following relations
\begin{enumerate}
\item\label{prop:chainHmin} For a quantum state $\rho_{ABC}$,
	\begin{eqnarray}\label{eq:chainHmin}
H_{\min}^{\epsilon}(A|BC)_{\rho}\geq  H_{\min}^{\frac{\epsilon}{4}}(AB|C)_\rho &-H_{\max}^{\frac{\epsilon}{4}}(B|C)_{\rho}\\
&\quad -2\log\de{1-\sqrt{1-\de{\frac{\epsilon}{4}}^2}}.\nonumber
	\end{eqnarray}
	\item\label{prop:chainClassicalregister} If $X$ is a classical register and $\rho_{ABX}$ a quantum-quantum-classical state, it holds that\footnote{In \cite{TomamichelBook} relation $H_{\min}^{\epsilon}(A|XB)_{\rho}\geq H_{\min}^{\epsilon}(A|B)_\rho-\log |X|$ was proved. Relation (\ref{eq:chainClassicalregister}) with the rank of $\rho_X$ follows as pointed out in Property~\ref{PropH}.(\ref{prop:Dimsystem}).}
	\begin{eqnarray}\label{eq:chainClassicalregister}
H_{\min}^{\epsilon}(A|XB)_{\rho}\geq H_{\min}^{\epsilon}(A|B)_\rho-\log \de{\rm{rank}(\rho_X)},
\end{eqnarray}
where $\rm{rank}(\rho_X)$ is the rank of state $\rho_X$.
\end{enumerate}
\end{prop}

A fully contained overview with properties and relations between different entropies can be found in \cite{TomamichelBook} (see also, \cite{EntropyZoo}).

\section{Security proof}

According to Definition~\ref{def:security}, a security proof of a DIQKD protocol consists in completeness and soundness. We start by proving completeness of Protocols \ref{prot:diqkd} and \ref{prot:diqkdIID}.

\begin{theorem}[Completeness]
The DIQKD protocols in consideration, Protocols \ref{prot:diqkd} and \ref{prot:diqkdIID} are $\epsilon_{DIQKD}^c$ complete, with
\begin{equation}
\epsilon_{DIQKD}^c\leq \epsilon_{EC}^c+\epsilon_{est}+\epsilon_{EC}.
\end{equation}
\end{theorem}

\begin{proof}
The protocols in consideration can abort in two steps. Either because the error correction fail, or because the estimated Bell violation is not high enough. Let us consider an honest implementation consisting of IID rounds where the expected winning CHSH probability is $\omega_{exp}$.
\begin{eqnarray*}
p(\rm{abort})&=p(\rm{(EC\, abort)\, {or}\, (EC\, does\, not\, abort \,and\, Bell\, test\, fail)} )\\
&\leq p({\rm{EC\, abort}})+p(\rm{EC\, does\, not\, abort \,and\, Bell\, test\, fail} )
\end{eqnarray*}
Now, the probability that the error correction protocol abort for an honest implementation is $p({\rm{EC\, abort}})\leq \epsilon_{EC}^c$. And for the other term we have
\begin{eqnarray*}
p({\rm{EC\, does\,}} & {\rm{  not\,abort  \,and\, Bell\, test\, fail}})\\
&=p(K_A=K_B)p(\sum_i C_i<\sum_i T_i \times \de{\omega_{exp}-\delta_{est}}|K_A=K_B)\\
&\quad+p(K_A\neq K_B)p(\sum_i C_i<\sum_i T_i \times \de{\omega_{exp}-\delta_{est}}|K_A\neq K_B)\\
&\leq \epsilon_{est}+\epsilon_{EC},
\end{eqnarray*}
where $\epsilon_{est}=\e^{-2\gamma n(\delta_{est}) ^2}$ follows from Hoeffding's inequality. 
\end{proof}

For the soundness proof we have to evaluate correctness and secrecy, Definitions \ref{def:correct} and \ref{def:secret}. For an error correction protocol with error parameter $\epsilon_{EC}$ we have that given that the error correction protocol does not abort, the probability that the string $\tilde{B}$ after error correction is equal to $A_1^n$ with probability higher than $1-\epsilon_{EC}$ and consequently
\begin{eqnarray}
P(K_A\neq K_B)\leq \epsilon_{EC}.
\end{eqnarray}

For the secrecy let us recall that, for each considered Protocol, $\Omega$ is defined as the event that the respective protocols do not abort. That happens when the error correction protocol does not abort and they achieved the required violation of CHSH according to Bob's estimation of Alice's string. Now, let us the define the event $\hat{\Omega}$ as the event $\Omega$ of the Protocol not aborting \textbf{and} the error correction being successful, \textit{i.e.} $\tilde{B}_1^n=A_1^n$. 
Now the quantity we need to estimate for the secrecy, relates to the event $\hat{\Omega}$ by
\begin{eqnarray}
{\|{\rho_{K_AE}}_{|\Omega}-{\tau_{K_A}\otimes \rho_E}\|}_1&\leq 
{\|{\rho_{K_AE}}_{|\Omega}-{\rho_{K_AE}}_{|\hat{\Omega}}\|}_1+{\|{\rho_{K_AE}}_{|\hat{\Omega}}-{\tau_{K_A}\otimes \rho_E}\|}_1 \nonumber\\
&\leq \epsilon_{EC}+{\|{\rho_{K_AE}}_{|\hat{\Omega}}-{\tau_{K_A}\otimes \rho_E}\|}_1
\end{eqnarray}
which follows from the fact that, since when error correction succeeds, the probability of $\tilde{B}_ 1^n=A_1^n$ is higher than $(1-\epsilon_{EC})$ then the following operator inequality holds: ${\rho_{K_AE}}_{|\Omega}\geq (1-\epsilon_{EC}){\rho_{K_AE}}_{|\hat{\Omega}}$.

In the following, we proceed to evaluate ${\|{\rho_{K_AE}}_{|\hat{\Omega}}-{\tau_{K_A}\otimes \rho_E}\|}_1$ in order to prove Theorems \ref{thm:rateEAT}, \ref{thm:rateIID} and \ref{thm:rateH2}.

\subsection{Proof of Theorem~\ref{thm:rateIID}}\label{Appendix:rIID}

In this Appendix we present the proof of Theorem~\ref{thm:rateIID}, that determines the size of a secret  key one can extract from Protocol~\ref{prot:diqkdIID} under the assumption that the eavesdropper is restricted to collective attacks. Importantly, Theorem~\ref{thm:rateIID} is based on the asymptotic equipartition  property, Theorem~\ref{thm:AEP}, in order to break the entropy of the $n$ rounds into the one-round entropy.

 The collective attacks assumption implies that in each round of the protocol the state distributed to Alice and Bob is the same, as well as their devices function in the same way, \textit{i.e.} the rounds are independent and identically distributed (IID). Therefore the state shared between Alice, Bob and Eve after Alice and Bob measure their raw keys is described by a tensor product form $\rho_{ABE}^{\otimes n}$.

The  asymptotic equipartition property (AEP) \cite{TCR09}, Theorem \ref{thm:AEP}, states that
the smooth min-entropy of a tensor product of states is almost equivalent (up to terms of order of $\sqrt{n}$) to $n$ times the von-Neumann entropy of an individual system. 
We now make use of the quantum AEP to derive the length of a secure key that one can achieve for Protocol~\ref{prot:diqkdIID}.

As established by the leftover hashing lemma, Theorem~\ref{thm:leftoverHmin}, the maximal length of a secure key is determined by the smooth min-entropy of Alice's raw key conditioned on all information available to the eavesdropper, given that the protocol did not abort. In the case of Protocol~\ref{prot:diqkdIID}, it is given by 
\begin{equation}\label{eq:Hforkey}
H_{\min}^{\frac{\epsilon_s}{p(\Omega)}}({A}_1^n|{X}_1^n{Y}_1^n{T}_1^nE O_{EC})_{\rho_{|\hat{\Omega}}}.
\end{equation}
Here we recall that $O_{EC}$ is the information exchanged by Alice and Bob during the error correction protocol. ${T}_1^n,{X}_1^n,{Y}_1^n$ are, respectively, the variable that determines whether the round is a test or a key generation round, and Alice and Bob's inputs, which are communicated publicly. $\hat{\Omega}$ is the event that error correction protocol succeeds, \textit{i.e.} $K_A=K_B$ and the CHSH probability estimated by Bob is $\omega\geq \omega_{exp}-\delta_{est}$.
In the following we describe the steps to estimate (\ref{eq:Hforkey}).

\gm{In order to avoid the conditioned state we can give one step back and note that in Definition~\ref{def:secret} we want to bound
\begin{eqnarray}
p(\hat{\Omega}){\|{\rho_{K_AHE|\hat{\Omega}}}-{\tau_{K_A}\otimes \rho_{HE|\hat{\Omega}}}\|}_1={\|{\rho_{K_AHE\land\hat{\Omega}}}-{\tau_{K_A}\otimes \rho_{HE\land\hat{\Omega}}}\|}_1
\end{eqnarray}
where $\rho_{K_AHE\land\hat{\Omega}}=p(\hat{\Omega})\rho_{K_AHE|\hat{\Omega}}$.
Now using the Leftover Hashing Lemma, Theorem~\ref{thm:leftoverHmin}, we can express an $(\epsilon_{PA}+\epsilon_s)$-secret key by
\begin{eqnarray}
l&=H_{\min}^{\epsilon_s}(A_1^n|E)_{\rho_{\land\Omega}}-2\log\de{\frac{1}{2\epsilon_{PA}}}.
\end{eqnarray}
Now we make use of the fact that the smooth-min-entropy of the conditioned state is lower bounded by the smooth-min-entropy of the state without conditioning, as proved in Ref.~\cite[Lemma 10]{RigorousQKD}
\begin{eqnarray}
H_{\min}^{\epsilon_s}(A_1^n|E)_{\rho_{\land\Omega}}\geq H_{\min}^{\epsilon_s}(A_1^n|E)_{\rho}.
\end{eqnarray}
In the following we proceed to estimate the quantity
\begin{equation}
H_{\min}^{\epsilon_s}({A}_1^n|{X}_1^n{Y}_1^n{T}_1^nE O_{EC})_{\rho}.
\end{equation}
}

\subsection*{\textbf{Step 1:} Accounting for the leakage in the error correction.}
 \setcounter{footnote}{0}
Using the chain rule relation for the smooth min-entropy conditioned on classical information, Property \ref{prop:chain}\textit{(\ref{prop:chainClassicalregister})}, we have
\begin{eqnarray}
\gm{H_{\min}^{\epsilon_s}({A}_1^n|{X}_1^n{Y}_1^n{T}_1^nE O_{EC})_{\rho}\geq H_{\min}^{\epsilon_s}({A}_1^n|{X}_1^n{Y}_1^n{T}_1^nE )_{\rho}-\mbox{leak}_{EC},}
\end{eqnarray}
where $\mbox{leak}_{EC}=\rm{rank}(\rho_{O_{EC}})$ represents the minimum amount of classical information that needs to be communicated from Alice to Bob in order to perform error correction\footnote{Note that in a realistic implementation Alice might send the error correction information using an encoding in order to overcome errors in the transmission due to channel losses. Therefore, in general $\rho_{O_{EC}}$ may not be full rank.}.
We consider that Alice and Bob use a protocol based on universal hashing which has minimum leakage \cite{BS93}. In \cite{RW05} it was proved that the minimum leakage is given by
\begin{eqnarray}
{\rm leak}_{EC}\leq H_0^{\epsilon'_{EC}}({A}_1^n|{B}_1^n{X}_1^n{Y}_1^n{T}_1^n)+\log \de{\frac{1}{\epsilon_{EC}}},
\end{eqnarray}
where, if Alice and Bob do not abort, then $K_A=K_B$ with probability at least $1-\epsilon_{EC}$. And for an honest implementation, the error correction protocol aborts with probability at most $\epsilon_{EC}^c=\epsilon'_{EC}+\epsilon_{EC}$. Here $H_0$ is a R\'enyi entropy first introduced in Ref.~\cite{RennerThesis} (in Ref.~\cite{TomamichelBook}, it is denoted $\bar{H}_0^{\uparrow}$).
The entropy $H_0^{\epsilon}$, relates to the smooth max-entropy in the following way \cite[Lemma 18]{TSSR11},
\begin{eqnarray}\label{eq:H0iid}
H_0^{\epsilon'_{EC}}({A}_1^n|{B}_1^n{X}_1^n{Y}_1^n{T}_1^n)\leq & H_{\max}^{\frac{\epsilon'_{EC}}{2}}({A}_1^n|{B}_1^n{X}_1^n{Y}_1^n{T}_1^n)\\
&+\log\de{\frac{8}{{\epsilon'}_{EC}^2}+\frac{2}{2-\epsilon'_{EC}}}.\nonumber
\end{eqnarray}  

We now can use of the asymptotic equipartition property, Theorem \ref{thm:AEP}, to decompose 
(\ref{eq:H0iid}) into the sum of the entropy of single rounds. Moreover,  for an honest implementation with winning CHSH probability $\omega_{exp}$ and QBER $Q$ we have that for  the 
test rounds  $H({A}|{B}{X}{Y}{T=1})=h(\omega_{exp})$ and for the key generation rounds 
$H({A}|{B}{X}{Y}{T=0})=h(Q)$. Therefore the one round entropy is given by
\begin{eqnarray}
H({A}|{B}{X}{Y}{T})&=p(T=0)H({A}|{B}{X}{Y}{T=0})+p(T=1)H({A}|{B}{X}{Y}{T=1})\nonumber \\
&=(1-\gamma)h(Q)+\gamma h(\omega_{exp}),
\end{eqnarray}
where in the first equality we have use Property~\ref{PropH}\textit{(\ref{prop:condClassic})}.

Therefore, the leakage due to error correction is given by
\begin{eqnarray}\label{eq:leakECiid}
{\rm leak}_{EC}\leq & n ((1-\gamma)h(Q)+\gamma h(\omega_{exp}))+\sqrt{n}\de{4\log\de{2\sqrt{2}+1}  \sqrt{\log\frac{8}{{\epsilon'}_{EC}^2}}}\nonumber\\
&+\log\de{\frac{8}{{\epsilon'}_{EC}^2}+\frac{2}{2-\epsilon'_{EC}}}+\log \de{\frac{1}{\epsilon_{EC}}}.
\end{eqnarray}

\gm{It is not known if an efficient error correction protocol can achieve the minimum leakage estimated in Eq.~(\ref{eq:leakECiid}), and practical implementations may use protocols with higher leakage. Ref.~\cite{TMPE17} analyses the leakage in error correction for concrete protocols, based on state-of-the-art error correcting codes, with efficient implementation. A more realistic analysis of the error correction leakage should take into account an specific code.}

\subsection*{\textbf{Step 2:} Breaking the entropy into single rounds. }

We now can use the asymptotic equipartition property in order to bound \gm{$H_{\min}^{\epsilon_s}({A}_1^n|{X}_1^n{Y}_1^n{T}_1^nE)_{\rho}$}. The assumption of collective attacks implies that the state under consideration has the tensor product form and therefore
\begin{eqnarray}
\gm{H_{\min}^{\epsilon_s}({A}_1^n|{X}_1^n{Y}_1^n{T}_1^nE)_{\rho}\geq n\, H(A|XYTE)_{\rho}-\sqrt{n}\delta(\epsilon_s,\eta),}
\end{eqnarray}
where $\delta(\epsilon_s,\eta)$ and $\eta$ are specified in Theorem \ref{thm:AEP}.

For the scenario under consideration we have
\begin{eqnarray}
\eta \leq 2\sqrt{2^{H_{\max}(A|XYTE)_{\rho}}}  +1\leq 2\sqrt{2}+1.
\end{eqnarray}
The first inequality follows from the fact that $A$ is a classical register and therefore has positive conditional min-entropy, which implies $-H_{\min}(A|XYTE)_{\rho}\leq H_{\min}(A|XYTE)_{\rho} \leq H_{\max}(A|XYTE)_{\rho}$. The second inequality follows from the fact that since $A$ is a binary variable $H_{\max}(A|XYTE)_{\rho}\leq 1$. Therefore,
\begin{eqnarray}
\delta(\epsilon_s, \eta)\leq 4\log\de{2\sqrt{2}+1}  \sqrt{\log\de{\frac{2}{\epsilon_s^2}}}.
\end{eqnarray}

\subsection*{\textbf{Step 3:} Estimating the one-round entropy.}

Now it only remains to lower bound \gm{$H(A|XYTE)_{\rho}$}. Lemma~\ref{lem:Hchsh} states the tight lower bound for the conditional von-Neumann entropy as a function of the winning probability  $\omega$ for the CHSH game derived in \cite{PAB09,DIEAT}. Using this bound we have that if $\rho $ is a state that achieves winning probability $\omega$ then
\begin{eqnarray}
\gm{H(A|XYTE)_{\rho}}&\geq 1-h\de{\frac{1}{2}+\frac{1}{2}\sqrt{16 \omega(\omega-1)+3}}.
\end{eqnarray}

Now, Protocol \ref{prot:diqkdIID} aborts if the observed frequency of winning events is smaller than $\omega_{exp}-\delta_{est}$. Therefore, given the event $\hat{\Omega}$ that Protocol \ref{prot:diqkdIID} does not abort and $K_A=K_B$, we have that Alice and Bob observe a violation higher than $\omega_{exp}-\delta_{est}$.  
Now we need to take into account that the CHSH violation is estimated with a finite number of rounds. So in order 
to infer the real winning probability $\omega^*$ of the IID implementation, we can make use of the Hoeffding's inequality in order to define a confidence interval: If $\omega^*< \omega_{exp}-\delta_{est}-\delta_{con}$ then
\begin{eqnarray}
Prob\de{\omega_{observed}\geq \omega_{exp}-\delta_{est}}\leq \e^{-2 \gamma n (\delta_{con})^2}:=\epsilon_{con}.
\end{eqnarray}
Therefore, given that Alice and Bob do not abort the protocol, we infer that the expected winning probability of the system under consideration is higher than  $\omega_{exp}-\delta_{est}-\delta_{con}$, and therefore
\begin{eqnarray}
\gm{H(A|XYTE)_{\rho}}\geq \\
\quad \quad \mbox{$1-h\Big(\frac{1}{2}+\frac{1}{2}\sqrt{16 (\omega_{exp}-\delta_{est}-\delta_{con})((\omega_{exp}-\delta_{est}-\delta_{con})-1)+3}\Big).$}\nonumber
\end{eqnarray}
\vspace{1em}

Putting the results of these steps together we have that 
either Protocol~\ref{prot:diqkdIID} aborts with probability higher than $1-(\epsilon_{con}+\epsilon_{EC})$, or the probability of aborting is smaller than $(\epsilon_{con}+\epsilon_{EC})$ and a $(2\epsilon_{EC}+\epsilon_s+\epsilon_{PA})$-correct-and-secret  key can be generated of size
\begin{eqnarray}
l\geq n& \Big[ 1-h\de{\frac{1}{2}+\frac{1}{2}\sqrt{16 (\omega_{exp}-\delta_{est}-\delta_{con})((\omega_{exp}-\delta_{est}-\delta_{con})-1)+3}}\nonumber\\
	& -(1-\gamma)h(Q)-\gamma h(\omega_{exp})\Big]\\
&\;\;\; -\sqrt{n}\de{4\log\de{2\sqrt{2}+1} \de{\sqrt{\log\frac{2}{\epsilon_s^2}}+ \sqrt{\log \frac{8}{{\epsilon'}_{EC}^2}}}}\nonumber\\
&\;\;\; -\log\de{\frac{8}{{\epsilon'}_{EC}^2}+\frac{2}{2-\epsilon'_{EC}} }-\log \de{\frac{1}{\epsilon_{EC}}}- 2\log\de{\frac{1}{2\epsilon_{PA}}}.\nonumber
\end{eqnarray}
This establishes Theorem~\ref{thm:rateIID}.

\subsection{Proof of Theorem~\ref{thm:rateH2}}\label{Appendix:rH2}

We now present the proof of Theorem~\ref{thm:rateH2}, that determines the size of a secret  key one can extract from Protocol~\ref{prot:diqkdIID} for collective attacks, but differently from Theorem~\ref{thm:rateIID}, we now use the additivity property of the 2-R\'enyi entropy, Lemma~\ref{lem:H2}, in order to break the entropy of the string into the one-round entropy.

We are now interested in estimate the length of a secure key as established in Theorem~\ref{thm:leftoverH2}, which is given by 
\begin{equation}\label{eq:H2forkey}
H_{2}^{\uparrow}({A}_1^n|{X}_1^n{Y}_1^n{T}_1^nE O_{EC})_{\rho_{| \hat{\Omega}}}.
\end{equation}

As in \ref{Appendix:rIID} we now present the steps that lead to the proof of Theorem~\ref{thm:rateH2}.

\subsection*{\textbf{Step 1:} Accounting for the leakage in the Error Correction.}

Using Property \ref{PropH}\textit{(\ref{prop:condClassic})}, we have
\begin{eqnarray}
H_{2}^{\uparrow}({A}_1^n|{X}_1^n{Y}_1^n{T}_1^nE O_{EC})_{\rho_{| \hat{\Omega}}}\geq H_{2}^{\uparrow}({A}_1^n|{X}_1^n{Y}_1^n{T}_1^nE )_{\rho_{| \hat{\Omega}}}-\mbox{leak}_{EC},
\end{eqnarray}
where $\mbox{leak}_{EC}=\rm{rank}(\rho_{O_{EC}})$ represents the minimum amount of classical information that needs to be communicated from Alice to Bob in order to perform error correction.

Now the error correction leakage $\mbox{leak}_{EC}$  is the same as derived in Equation (\ref{eq:leakECiid}).

\subsection*{\textbf{Step 2:} Breaking the entropy into single rounds. }

\gm{We first use Property~\ref{PropH}\textit{(\ref{prop:condClassic})} in order to express the entropy of the state conditioned on the event $\hat{\Omega}$ in terms of the entropy of the unconditioned state
\begin{eqnarray}
H_{2}^{\uparrow}({A}_1^n|{X}_1^n{Y}_1^n{T}_1^nE )_{\rho_{| \hat{\Omega}}}\geq H_{2}^{\uparrow}({A}_1^n|{X}_1^n{Y}_1^n{T}_1^nE )_{\rho}-2\log\de{\frac{1}{p_{\hat{\Omega}}}}.
\end{eqnarray}
}

We can now make use the additivity property of $2$-R\'eyni entropy, Lemma~\ref{lem:H2}, in order to bound \gm{$H_{2}^{\uparrow}({A}_1^n|{X}_1^n{Y}_1^n{T}_1^nE )_{\rho}$}. The assumption of collective attacks implies that the state under consideration has the tensor product form and therefore
\begin{eqnarray}
\gm{H_{2}^{\uparrow}({A}_1^n|{X}_1^n{Y}_1^n{T}_1^nE)_{\rho}\geq n\, H_2(A|XYTE)_{\rho}},
\end{eqnarray}
where now the single round entropy in consideration is the conditional collision entropy.

\subsection*{\textbf{Step 3:} Estimating the one-round entropy.}

Now it only remains to lower bound \gm{$H_2(A|XYTE)_{\rho}$}. 
Theorem~\ref{thm:H2} shows that a tight lower bound for the conditional collision entropy as a function of the violation $\beta$ coincides with the previously derived conditional min-entropy\cite{HminBell}, eq.(\ref{eq:Hminbeta}).
\gm{Therefore, for a state $\rho$ that wins the CHSH game with probability $\omega$}
\begin{eqnarray}
\gm{H_2(A|XYTE)_{\rho}}&\geq-\log\de{\frac{1}{2}+\frac{1}{2}\sqrt{16\omega(1-\omega)-2}}.
\end{eqnarray}

\gm{Now, either the expected winning probability of the system under consideration is smaller than  $\omega_{exp}-\delta_{est}-\delta_{con}$, in which case the protocol aborts with probability higher than $1-(\epsilon_{con}+\epsilon_{EC})$, or $p_{\hat{\Omega}}>\epsilon_{con}+\epsilon_{EC}$ which implies that the system has winning probability larger than $\omega_{exp}-\delta_{est}-\delta_{con}$, and  }
\begin{eqnarray}
\gm{H_2(A|XYTE)_{\rho}}\geq \\
\quad \quad -\log\de{\frac{1}{2}+\frac{1}{2}\sqrt{16(\omega_{exp}-\delta_{est}-\delta_{con})(1-(\omega_{exp}-\delta_{est}-\delta_{con}))-2}}.\nonumber
\end{eqnarray}

In conclusion we have that,  
either Protocol~\ref{prot:diqkdIID} aborts with probability higher than $1-(\epsilon_{con}+\epsilon_{EC})$, or the probability of not aborting is greater than $(\epsilon_{con}+\epsilon_{EC})$ and a $(2\epsilon_{EC}+\epsilon_{PA})$-correct-and-secret  key is generated of size:
\begin{eqnarray}
l\geq n& \Big[ -\log\de{\frac{1}{2}+\frac{1}{2}\sqrt{16(\omega_{exp}-\delta_{est}-\delta_{con})(1-(\omega_{exp}-\delta_{est}-\delta_{con}))-2}}\nonumber\\
	& -(1-\gamma)h(Q)-\gamma h(\omega_{exp})\Big]\\
&\;\;\; -\sqrt{n}\de{4\log\de{2\sqrt{2}+1}  \sqrt{\log \frac{8}{{\epsilon'}_{EC}^2}}}\nonumber\\
&\;\;\; -\log\de{\frac{8}{{\epsilon'}_{EC}^2}+\frac{2}{2-\epsilon'_{EC}} }\\ \nonumber
&-\log \de{\frac{1}{\epsilon_{EC}}}- 2\log\de{\frac{1}{2\epsilon_{PA}}} \gm{-2\log\de{\frac{1}{\epsilon_{con}+\epsilon_{EC}}}}.\nonumber
\end{eqnarray}
This establishes Theorem~\ref{thm:rateH2}.

\subsection{Proof of Theorem~\ref{thm:rateEAT}}\label{Appendix:rEAT}

In this Appendix we present the proof of Theorem~\ref{thm:rateEAT}, which establishes the size of a secure key that can be extracted from Protocol~\ref{prot:diqkd} for general coherent attacks. We follow closely the proof developed in \cite{EATpublish,DIEAT}.

In  Protocol~\ref{prot:diqkd}, the number of rounds is not fixed. Instead, 
Protocol \ref{prot:diqkd} has a fixed number of blocks $m$, such that the maximum number of rounds inside a block is set to $s_{\max}=\left\lceil \frac{1}{\gamma} \right\rceil$. 
This is a technicality introduced in \cite{DIEAT,EATpublish} in order to get a better pre-factor for the overhead terms that scale with $\sqrt{n}$. 
For each block $j$ Alice and Bob run the protocol until they have a test round or they reach the maximum number of rounds $s_{\max}$. At each round $j_i$ Alice and Bob choose  a random bit $T_{j_i}$, such that $P(T_{j_i}=1)=\gamma$, which determines whether they are going to test the CHSH inequality or make a key generation round. They repeat the process until they obtain $T_{j_i}=1$ or $i=s_{\max}$.
With these constraints the expected number of rounds in a block is given by
\begin{eqnarray}\label{eq:s_exp}
\bar{s}=\frac{1-(1-\gamma)^{\left\lceil \frac{1}{\gamma} \right\rceil}}{\gamma},
\end{eqnarray}
and the expected number of rounds is
\begin{eqnarray}\label{eq:n_exp}
{n}=m\bar{s}.
\end{eqnarray}
For details on the derivation of equations (\ref{eq:s_exp}) and (\ref{eq:n_exp}) see Ref.~\cite[Appendix B]{DIEAT}

We now proceed to derive the key rates against a general coherent attack. In order to calculate the size of the key we need to estimate
\begin{eqnarray}
H_{\min}^{\frac{\epsilon_s}{p({\Omega})}}(\vec{A}_1^m|\vec{X}_1^m\vec{Y}_1^m\vec{T}_1^mE O)_{\rho_{|\hat{\Omega}}}.
\end{eqnarray}
Now $\vec{A}_1^m$ denotes the total string of bits, expected to be of size $n$, and $\vec{A}_i$ denotes the string of outputs generated in the block $i$, and similarly for the other variables.  
In the following, we proceed step by step in  order to lower bound $H_{\min}^{\frac{\epsilon_s}{p({\Omega})}}(\vec{A}_1^m|\vec{X}_1^m\vec{Y}_1^m\vec{T}_1^mE O)_{\rho_{|\hat{\Omega}}}$ and we detail the changes introduced to the original analysis \cite{DIEAT,EATpublish}.

\subsection*{\textbf{Step 1:} Accounting for the leakage in the error correction.}

Similar to the proof of Protocol~\ref{prot:diqkdIID}, we have that
\begin{eqnarray}
H_{\min}^{\frac{\epsilon_s}{p({\Omega})}}(\vec{A}_1^m|\vec{X}_1^m\vec{Y}_1^m\vec{T}_1^mE O)_{\rho_{| \hat{\Omega}}}\geq H_{\min}^{\frac{\epsilon_s}{p({\Omega})}}(\vec{A}_1^m|\vec{X}_1^m\vec{Y}_1^m\vec{T}_1^mE )_{\rho_{| \hat{\Omega}}}-\mbox{leak}_{EC},
\end{eqnarray}
and
\begin{eqnarray}
\mbox{leak}_{EC}&\leq H_0^{\epsilon'_{EC}}(\vec{A}_1^m|\vec{B}_1^m\vec{X}_1^m\vec{Y}_1^m\vec{T}_1^m)+\log\de{\frac{1}{\epsilon_{EC}}}\\
&\leq H_{\max}^{\frac{\epsilon'_{EC}}{2}}(\vec{A}_1^m|\vec{B}_1^m\vec{X}_1^m\vec{Y}_1^m\vec{T}_1^m)\\
&\quad\quad +\log\de{\frac{8}{{\epsilon'_{EC}}^2}+\frac{2}{(2-\epsilon'_{EC})}}+\log\de{\frac{1}{\epsilon_{EC}}}.\nonumber
\end{eqnarray}
However, now we need to take into account for the fact that the number of rounds in the protocol is not fixed.
Following the steps of Ref. \cite{DIEAT}, we first note that the number of rounds $N$ obtained in an implementation of the Protocol~\ref{prot:diqkd}  satisfies:
\begin{eqnarray}
P[N\geq {n}+t]\leq \exp\de{-\frac{2t^2\gamma^2}{m(1-\gamma)^2}}:=\epsilon_t,
\end{eqnarray}
where $n=m\bar{s}$ is the expected number of rounds and  $t=\sqrt{-\frac{m(1-\gamma)^2 \log \epsilon_t}{2 \gamma^2}}$.
Moreover, by the definition of smooth max-entropy one have that 
\begin{eqnarray}
 H_{\max}^{\epsilon}(\vec{A}_1^m|\vec{B}_1^m\vec{X}_1^m\vec{Y}_1^m\vec{T}_1^mN)\leq H_{\max}^{\epsilon-\sqrt{\epsilon_t}}(\vec{A}_1^m|\vec{B}_1^m\vec{X}_1^m\vec{Y}_1^m\vec{T}_1^mN\leq n+t).
\end{eqnarray}
Note that $N$ can be included in the entropy since it is completely determined by $\vec{T}_1^m$.

Now applying the asymptotic equipartition property, Theorem~\ref{thm:AEP}, to the maximal length $N=n+t$ we have
\begin{eqnarray}
\mathrm{leak}_{EC}&\leq (\bar{n}+t)\cdot\De{(1-\gamma)h(Q)+\gamma h(\omega_{exp})} \nonumber\\
&\quad+\sqrt{\bar{n}+t}\, \nu_2 +\log\de{\frac{8}{{\epsilon'_{EC}}^2}+\frac{2}{(2-\epsilon'_{EC})}}+\log\de{\frac{1}{\epsilon_{EC}}}, \nonumber
\end{eqnarray}
where $\nu_2=4\log\de{2\sqrt{2}+1}\sqrt{2\log\de{\frac{8}{\de{\epsilon'_{EC}-2\sqrt{\epsilon_t}}^2}}}$ and $\epsilon_t$ is a free parameter to be optimised.

If the error correction protocol does not abort, then
\begin{eqnarray}
P(K_A\neq K_B)\leq \epsilon_{EC}.
\end{eqnarray}
And the completeness of the error correction protocol (\ie , the probability of not aborting in an honest IID implementation) is given by $\epsilon_{EC}^c=\epsilon'_{EC}+\epsilon_{EC}$.

\subsection*{\textbf{Step 2:} Chain rule.}

In Protocol~\ref{prot:diqkd}, a statistical test is performed on the variable $C_i$ which accounts for the condition of winning the CHSH game being satisfied or not. In order to use the entropy accumulation theorem, we need to be able to infer the value of this variable $C_i$ from the variables that appear in the smooth min-entropy we are calculating. 

Here we choose to use a chain rule, relation (\ref{eq:chainHmin}), with the variable $C_i$ itself, as opposed to using the variable $B_i$ as is done in \cite{DIEAT}. The reason is that the dimension of the variable $C_i$ is smaller than $B_i$, as for each block the variable $C_i$ assumes one out of three values. This leads to a slight improvement in rates achieved in the finite regime:  
\begin{eqnarray}
H_{\min}^{\frac{\epsilon_s}{p({\Omega})}}(\vec{A}_1^m|\vec{X}_1^m\vec{Y}_1^m\vec{T}_1^mE)_{\rho_{| \hat{\Omega}}}&\geq H_{\min}^{\frac{\epsilon_s}{4p({\Omega})}}(\vec{A}_1^m{C}_1^m|\vec{X}_1^m\vec{Y}_1^m\vec{T}_1^mE )_{\rho_{| \hat{\Omega}}}\nonumber\\
&\quad- H_{\max}^{\frac{\epsilon_s}{4p({\Omega})}}({C}_1^m|\vec{A}_1^m\vec{X}_1^m\vec{Y}_1^m\vec{T}_1^mE)_{\rho_{| \hat{\Omega}}}\\
&\quad \quad-3 \log\de{1-\sqrt{1-\de{\frac{\epsilon_s}{4p({\Omega})}}^2}} \nonumber\\
&\geq H_{\min}^{\frac{\epsilon_s}{4p({\Omega})}}(\vec{A}_1^m{C}_1^m|\vec{X}_1^m\vec{Y}_1^m\vec{T}_1^mE )_{\rho_{| \hat{\Omega}}}\nonumber\\
&\quad- H_{\max}^{\frac{\epsilon_s}{4p({\Omega})}}({C}_1^m|\vec{T}_1^mE)_{\rho_{| \hat{\Omega}}} \label{eq.rand1}\\
&\quad\quad-3 \log\de{1-\sqrt{1-\de{\frac{\epsilon_s}{4(\epsilon_{EA}+\epsilon_{EC})}}^2}}.\nonumber
\end{eqnarray}
In inequality (\ref{eq.rand1}) we use the fact that $p({\Omega})\geq (\epsilon_{EA}+\epsilon_{EC})$ and that removing the conditioning on classical variables can only increase the entropy, which can be seen as a particular case of data processing, Property \ref{PropH}\textit{(\ref{propDataProcess})}.

\subsection*{\textbf{Step 3:} Upper bound on $H_{\max}^{\frac{\epsilon_s}{4p({\Omega})}}({C}_1^m|\vec{T}_1^mE)_{\rho_{|\hat{\Omega}}}$.}

We can use the entropy accumulation theorem to upper bound $H_{\max}^{\frac{\epsilon_s}{4p({\Omega})}}({C}_1^m|\vec{T}_1^mE )_{\rho_{|\hat{\Omega}}}$.
In order to do that we only have to find a max-tradeoff function for a protocol with $m$ rounds.
We have that for any distribution $\vec{{p}}=({p}(1),{p}(0),{p}(\perp))$ of the variable $C$:
\begin{eqnarray}
H(C_i|\vec{T}_iE)_{\rho_{|\hat{\Omega}}}&= p(\vec{T}_i=\vec{0})H(C_i|\vec{T}_i=\vec{0}E)_{\rho_{|{\hat{\Omega}}}}\\
&\quad \quad \quad + p(\vec{T}_i\neq \vec{0})H(C_i|\vec{T}_i\neq \vec{0} E)_{\rho_{|{\hat{\Omega}}}} \nonumber\\
&= p(\vec{T}_i\neq \vec{0})H(C_i|\vec{T}_i\neq \vec{0} E)_{\rho_{|{\hat{\Omega}}}} \label{eq.rand2}\\
&\leq h\de{\frac{{p}(1)}{1-{p}(\perp)}}=h\de{\frac{{p}(1)}{1-(1-\gamma)^{s_{\max}}}}=h(\omega),\label{eq.rand3}
\end{eqnarray}
where in (\ref{eq.rand2}) we use the fact that $H(C_i|\vec{T}_i=\vec{0}E)=0$, and in (\ref{eq.rand3}) we use that $p(\vec{T}_i\neq \vec{0})\leq 1$ and that $\frac{{p}(1)}{1-(1-\gamma)^{s_{\max}}}\equiv \omega$.
Note that $h(\cdot)$ is a concave function.

Now we can take $f_{\max}= h(\omega_{exp}-\delta_{est})$ and $\|\nabla f_{\max}\|_{\infty}= \frac{1}{1-(1-\gamma)^{s_{\max}}}\times \frac{\partial h}{\partial \omega}\Bigr|_{\omega_{exp}-\delta_{est}}$, where $\omega_{exp}$ is the expected winning probability of the CHSH game in an honest implementation and $\delta_{est}$ accounts for the statistical confidence interval of the experiment. Using the entropy accumulation theorem, Theorem \ref{thm.EAT}, we have
\begin{eqnarray}
H_{\max}^{\frac{\epsilon_s}{4p({\Omega})}}(C_1^m|\vec{T}_1^mE)_{\rho_{| \hat{\Omega}}}&\leq m\, h(\omega_{exp}-\delta_{est})+\sqrt{m}\, \nu_1
\end{eqnarray}
where 
\begin{eqnarray}
\nu_1=2 \de{\log 7 +\left\lceil\frac{|h'(\omega_{exp}+\delta_{est})|}{1-(1-\gamma)^{s_{\max}}}\right\rceil}\sqrt{1-2\log\epsilon_s},
\end{eqnarray}
and $h'$ represents the derivative of the binary entropy function, $\frac{\partial h(\omega)}{\partial \omega}$.

\subsection*{\textbf{Step 4:} Lower bound on $H_{min}^{\frac{\epsilon_s}{p({ \Omega})}}(\vec{A}_1^mC_1^m|\vec{X}_1^m\vec{Y}_1^m\vec{T}_1^mE )_{\rho_{|\hat{\Omega}}}$.}
Finally, we apply the entropy accumulation theorem to lower bound the term $H_{min}^{\frac{\epsilon_s}{p({ \Omega})}}(\vec{A}_1^mC_1^m|\vec{X}_1^m\vec{Y}_1^m\vec{T}_1^mE )_{\rho_{| \hat{\Omega}}}$.
Therefore we need to find a min-tradeoff function such that
\begin{eqnarray}
f_{\min}(\vec{q})\leq \inf_{\sigma_{R_{j-1}E}:\mathcal{M}_j(\sigma)_{C_j}=\vec{q}}H(\vec{A}_j{C}_j|\vec{X}_j\vec{Y}_j\vec{T}_jE)_{\mathcal{M}_j(\sigma)}
\end{eqnarray}

Note that the length of each block is variable. However, we can consider that all the blocks have size $s_{\max}$ and set all the variables to $\bot$ for the rounds which are not performed.

First note that 
\begin{eqnarray}
H(\vec{A}_j{C}_j|\vec{X}_j\vec{Y}_j\vec{T}_jE)\geq H(\vec{A}_j|\vec{X}_j\vec{Y}_j\vec{T}_jE).
\end{eqnarray} 
And from now on, we follow the same steps as Ref. \cite{DIEAT}.

Using the chain-rule for Von Neuman, Property~\ref{PropH}\textit{(\ref{prop:condClassic})}, entropy we have
\begin{eqnarray}
H(\vec{A}_j|\vec{X}_j\vec{Y}_j\vec{T}_jE)=\sum_{i=1}^{s_{\max}}H(A_{j,i}|\vec{X}_j\vec{Y}_j\vec{T}_jE{A_j}_{1}^{i-1}).
\end{eqnarray} 
and for every $i \in [s_{\max}]$, 
\begin{eqnarray}
H&(A_{j,i}|\vec{X}_j\vec{Y}_j\vec{T}_jE{A_j}_{1}^{i-1})= \nonumber\\
&\quad\quad \quad = p({T_j}_{1}^{i-1}=\vec{0})H(A_{j,i}|\vec{X}_j\vec{Y}_jE{A_j}_{1}^{i-1}{T_j}_{i}^{s_{\max}},{T_j}_{1}^{i-1}=\vec{0})\\
& \quad\quad \quad \quad+p({T_j}_{1}^{i-1}\neq\vec{0})H(A_{j,i}|\vec{X}_j\vec{Y}_jE{A_j}_{1}^{i-1}{T_j}_{i}^{s_{\max}},{T_j}_{1}^{i-1}\neq\vec{0}) \nonumber\\
&\quad \quad \quad=(1-\gamma)^{(i-1)}H(A_{j,i}|\vec{X}_j\vec{Y}_jE{A_j}_{1}^{i-1}{T_j}_{i}^{s_{\max}},{T_j}_{1}^{i-1}=\vec{0}),
\end{eqnarray}
where we used the fact that $H(A_{j,i}|\vec{X}_j\vec{Y}_jE{A_j}_{1}^{i-1}{T_j}_{i}^{s_{\max}},{T_j}_{1}^{i-1}\neq\vec{0})=0$.
Therefore
\begin{eqnarray}
H&(\vec{A}_j|\vec{X}_j\vec{Y}_j\vec{T}_jE)=\\
&\quad \quad \quad \quad \sum_{i=1}^{s_{max}}(1-\gamma)^{(i-1)}H(A_{j,i}|\vec{X}_j\vec{Y}_jE{A_j}_{1}^{i-1}{T_j}_{i}^{s_{\max}},{T_j}_{1}^{i-1}=\vec{0}).\nonumber
\end{eqnarray}
Each term $H(A_{j,i}|\vec{X}_j\vec{Y}_jE{A_j}_{1}^{i-1}{T_j}_{i}^{s_{\max}},{T_j}_{1}^{i-1}=\vec{0})$ can be seen as the entropy of a single round. An expression for the entropy of a single round was derived for collective attacks in \cite{PAB09}. This gives us:
\begin{eqnarray}\label{eq.Hblock}
H&(\vec{A}_j{C}_j|\vec{X}_j\vec{Y}_j\vec{T}_jE)=\\
&\quad \quad \quad \quad \sum_{i=1}^{s_{max}}(1-\gamma)^{(i-1)}\De{1-h\de{\frac{1}{2}+\frac{1}{2}\sqrt{16\omega_i(\omega_i-1)+3}}}\nonumber
\end{eqnarray}
such that
\begin{eqnarray}
{p}(1)=\sum_{i=1}^{s_{max}}\gamma(1-\gamma)^{(i-1)}\omega_i.
\end{eqnarray}
Now, in \cite{DIEAT} it is proved that the minimum of (\ref{eq.Hblock}) is achieved for 
\begin{eqnarray}
\omega_i^*=\frac{{p}(1)}{1-(1-\gamma)^{s_{\max}}}\quad \forall i,
\end{eqnarray}
and therefore we have a min-tradeoff function:
\begin{eqnarray}\label{eq.mintradeoff}
     \mbox{$ g({\vec{p}}) = 
 {s}\De{1-h\de{\frac{1}{2}+\frac{1}{2}\sqrt{16\frac{{p}(1)}{1-(1-\gamma)^{s_{max}}}\de{\frac{{p}(1)}{1-(1-\gamma)^{s_{max}}} -1}+3 } }}$}; 
\end{eqnarray}
for $\frac{{p}(1)}{1-(1-\gamma)^{s_{max}}} \in \De{\frac{3}{4},\frac{2+\sqrt{2}}{4}}$. 

 \setcounter{footnote}{0}
Note that as ${p}(1)\to  ((1-(1-\gamma)^{s_{\max}})\frac{2+\sqrt{2}}{4}$, the gradient of $g(\vec{p})$ tends to infinity, which compromises the $\sqrt{n}$ term that depends on the norm of the gradient of $f$. Since $g(\vec{p})$ is a convex function,  the tangent line in any point $\vec{p}_t$ is a lower bound to $g(\vec{p})$. Therefore, as in \cite{DIEAT,EATpublish}, we take the min-tradeoff function to be a tangent $g$ in a point $\vec{p}_t$ to be optimized\footnote{In \cite{DIEAT,EATpublish} the authors consider the following min-tradeoff function 
\begin{equation*}
f_{\min}(\vec{p})= \left\{\begin{array}{@{}l@{\quad}l}
      g(\vec{p}) & \mbox{if $p_t(1)>p(1)$}\\[\jot]
      F_{\min}(\vec{p},\vec{p}_t)= & \mbox{if $p_t(1)\leq p(1)$}
    \end{array}.\right.
\end{equation*} 
We remark that, since the gradient of $ g(\vec{p})$ is an increasing function of $p(1)$, the optimum value for $\eta_{opt}$ is always achieved for $p_t(1)\leq p(1)$. }:
\begin{eqnarray}\label{eq.fmin}
F_{\min}({p},{p}_t)=
\frac{d}{d {p}(1)}g({p}) \Big|_{\tilde{p}_t}\cdot {p}(1)+\de{ g({p}_t)- \frac{d}{d{p}(1)}g({p})\Big|_{{p}_t}\cdot {p}_t(1) }.   
\end{eqnarray}
Then we have
\begin{eqnarray}
H_{min}^{\frac{\epsilon_s}{4p({\Omega})}}(\vec{A}_1^mC_1^m|\vec{X}_1^m\vec{Y}_1^m\vec{T}_1^mE )_{\rho_{|\hat{\Omega}}}&>m\cdot \eta_{opt}=\frac{\bar{n}}{\bar{s}}\cdot \eta_{opt},
\end{eqnarray}
where
\begin{eqnarray}
\eta_{opt}=\max_{\frac{3}{4}<\frac{\tilde{p}_t(1)}{1-(1-\gamma)^{s_{max}}}<\frac{2+\sqrt{2}}{4}} \De{F_{\min}(\tilde{p},\tilde{p}_t)-\frac{1}{\sqrt{m}}\nu_3},
\end{eqnarray}
such that
\begin{eqnarray}
\nu_3 =2 \de{\log\de{1+2\cdot 2^{s_{\max}}3}+\left\lceil \frac{d}{d{p}(1)}g(\tilde{p})\big|_{{p}_t}\right\rceil}\sqrt{1-2\log \epsilon_s}.
\end{eqnarray}

Finally, the length of a secure key that can be extracted is given by
\begin{eqnarray}
l\geq& \frac{\bar{n}}{\bar{s}}\eta_{opt} -\frac{\bar{n}}{\bar{s}}h(\omega_{exp}-\delta_{est}) -\sqrt{\frac{\bar{n}}{\bar{s}}}\nu_1 \nonumber \\
& -(\bar{n}+t)\cdot\De{(1-\gamma)h(Q)+\gamma h(\omega_{exp})}\\
&-\sqrt{\bar{n}+t}\, \nu_2 -\log\de{\frac{8}{{\epsilon'_{EC}}^2}+\frac{2}{(2-\epsilon'_{EC})}}-\log\de{\frac{1}{\epsilon_{EC}}} \nonumber\\
&-3\log\de{1-\sqrt{1-\de{\frac{\epsilon_s}{4}}^2}}-2\log\de{\frac{1}{2\epsilon_{PA}}}.\nonumber
\end{eqnarray}

\section{Proof of Theorem~\ref{thm:H2}}\label{Appendix:H2}


\begin{thm6}
There exist a state $\rho^*_{AB}$ and measurements for Alice and Bob such that, $\rho^*_{AB}$ achieves violation $\beta$ and the collision entropy of Alice's output $A$ conditioned on Eve's quantum information $E$ is 
\begin{equation}
H_2(A|E)_{\rho^*}= -\log\de{\frac{1}{2} +\frac{1}{2}\sqrt{2-\frac{\beta^2}{4}}}.
\end{equation}
\end{thm6}

\begin{proof}
The proof consists in exhibiting a state $\rho^*_{AB}$ and measurements for Alice and Bob such that the lower bound given by eq.(\ref{eq:Hminbeta}) is saturated. Our derivation is based on the techniques presented in Ref.~\cite{PAB09}, which led to a tight lower bound for the conditional von-Neumann entropy.

 Let us consider that Alice and Bob share a Bell diagonal state $\rho_{AB}$
 \begin{equation}
 \rho_{AB}=\lambda_{00} \Phi_{00}+\lambda_{01}\Phi_{01}+\lambda_{10}\Phi_{10}+\lambda_{11}\Phi_{11}
 \end{equation}
 where $\Phi_{ij}=\ketbra{\Phi_{ij}}{\Phi_{ij}}$ and $\ket{\Phi_{ij}}=I\otimes X^iZ^j \de{\frac{1}{\sqrt{2}}(\ket{00}+\ket{11})}$.
We first prove the following result:
  \begin{lemma}
 For a Bell-diagonal state where Alice performs a measurement in the $Z$-basis we have that
 \begin{eqnarray}
 H_2(A|XYE)_{\rho}\geq-\log \de{\frac{1}{2} + \sqrt{\lambda_{00}\lambda_{01}}+\sqrt{\lambda_{11}\lambda_{10}}}.
 \end{eqnarray}
 \end{lemma}
 
 \begin{proof}
 Given a Bell diagonal state $\rho_{AB}(\lambda_{00},\lambda_{01},\lambda_{10},\lambda_{11})$, a purification $\ket{\Psi}_{ABE}$ of this state is given by
\begin{eqnarray}
\ket{\Psi}_{ABE}=&\sqrt{\lambda_{00}}\ket{\Phi_{00}}_{AB}\ket{e_1}_E+\sqrt{\lambda_{01}}\ket{\Phi_{01}}_{AB}\ket{e_2}_E\\
&+\sqrt{\lambda_{10}}\ket{\Phi_{10}}_{AB}\ket{e_3}_E+\sqrt{\lambda_{11}}\ket{\Phi_{11}}_{AB}\ket{e_4}_E.\nonumber
\end{eqnarray}

After Alice measures in the Z basis we have
\begin{equation}\label{eq:rhoAE}
\rho_{AE}=\frac{1}{2}\ketbra{0}{0}\otimes \rho_{E|0}+\frac{1}{2}\ketbra{1}{1}\otimes \rho_{E|1}
\end{equation}
where
\begin{eqnarray}
\rho_{E|0}=\ketbra{\psi_1}{\psi_1}+\ketbra{\psi_2}{\psi_2}\; {\rm \,and\,} \;\rho_{E|1}=\ketbra{\psi_3}{\psi_3}+\ketbra{\psi_4}{\psi_4},
\end{eqnarray}
with non-normalized states
\begin{eqnarray*}
\ket{\psi_1}&=\de{\sqrt{\lambda_{00}}\ket{e_1}+\sqrt{\lambda_{01}}\ket{e_2}},\\
\ket{\psi_2}&=\de{\sqrt{\lambda_{10}}\ket{e_3}+\sqrt{\lambda_{11}}\ket{e_4}},\\
\ket{\psi_3}&=\de{\sqrt{\lambda_{10}}\ket{e_3}-\sqrt{\lambda_{11}}\ket{e_4}},\\
\ket{\psi_4}&=\de{\sqrt{\lambda_{00}}\ket{e_1}-\sqrt{\lambda_{01}}\ket{e_2}}.
\end{eqnarray*}

The collision entropy of a cq-state $\rho_{AE}$ is given by
\begin{equation}
H_2(A|E)_{\rho}=-\log \Tr \de{\rho_E^{-1/2}\rho_{AE}\rho_E^{-1/2}\rho_{AE}},
\end{equation}
which, evaluated for the state (\ref{eq:rhoAE}) gives
\begin{eqnarray}
H_2(A|E)_{\rho}=&-\log \left(\frac{1}{2} +\de{\sqrt{\lambda_{00}}\sqrt{\lambda_{01}}+\sqrt{\lambda_{10}}\sqrt{\lambda_{11}}}\right).\nonumber
\end{eqnarray}
\end{proof}

Now let us consider a Bell diagonal state $\rho^*_{AB}$ such that
\begin{eqnarray}\label{eq:rhostar}
\lambda_{00}=&R\cos \theta, \quad \lambda_{01}=R\sin \theta,\quad \lambda_{10}=\lambda_{11}=0, \\
&\quad\quad \mbox{s.t.}\quad \cos \theta +\sin \theta =\frac{1}{R}\nonumber
\end{eqnarray}
which can hold for $R>\frac{1}{\sqrt{2}}$. This choice is inspired by the optimal strategy that maximizes the conditional von Neumann entropy as shown in \cite{PAB09}.

For these parameters we have that 
 \begin{eqnarray}
 H_2(A|XYE)_{\rho^*}\geq-\log \de{\frac{1}{2} +R\sqrt{\frac{1}{2}\de{\frac{1}{R^2}-1}}}
  \end{eqnarray}

Finally, we know from \cite{HHH95} that for a state $\rho_{AB}(\lambda_{00},\lambda_{01},\lambda_{10},\lambda_{11})$, the maximal violation $\beta_{\max}$ of the CHSH inequality is given by
\begin{eqnarray}
\beta_{\max}=\max \Big\{2\sqrt{2}&\sqrt{(\lambda_{00}-\lambda_{11})^2+(\lambda_{01}-\lambda_{10})^2},\\
& 2\sqrt{2}\sqrt{(\lambda_{00}-\lambda_{10})^2+(\lambda_{01}-\lambda_{11})^2}\Big\}\nonumber
\end{eqnarray}
and that this violation can be achieved with one of Alice's measurement being in the Z basis.

Therefore, for the state $\rho^*_{AB}$, specified by (\ref{eq:rhostar}), and Alice and Bob performing the measurements that gives the maximum violation achievable for the CHSH inequality, we have that $\beta=2\sqrt{2}R$. This implies 
 \begin{eqnarray}\label{eq:H2beta}
 H_2(A|XYE)_{\rho^*}=-\log \de{\frac{1}{2} + \frac{1}{2}\sqrt{2-\frac{\beta}{4}}}.
 \end{eqnarray}
\end{proof}

 \end{document}